\RequirePackage{fix-cm}
\documentclass[smallextended]{svjour3}       
\usepackage{graphicx}

\usepackage{cprotect}
\usepackage[usenames, dvipsnames]{color}
\usepackage[colorlinks=true,linkcolor=blue,citecolor=blue,urlcolor=blue,pagebackref]{hyperref}
\usepackage{amssymb,amsmath}
\usepackage{longtable}

\usepackage{bm}

\newcommand{\1}{{\rm 1}\kern-0.24em{\rm I}}

\newcommand{\noteJJL}[1]{\textcolor{black}{#1}}
\newcommand{\noteW}[1]{\textcolor{black}{#1}}
\newcommand{\noteJJLnew}[1]{\textcolor{black}{#1}}
\newcommand{\JJL}[1]{\textcolor{black}{#1}}

\newcommand{\JL}[1]{\textcolor{black}{#1}}

\begin{document}

\title{TROM: A Testing-Based Method for Finding Transcriptomic Similarity of Biological Samples
}

\titlerunning{TROM: Transcriptome Overlap Measure}        

\author{Wei Vivian Li         \and
        Yiling Chen  \and
        Jingyi Jessica Li*
}


\institute{Wei Vivian Li \at
              Department of Statistics, University of California, Los Angeles, CA 90095-1554, USA \\
           \and
           Yiling Chen \at
              Department of Statistics, University of California, Los Angeles, CA 90095-1554, USA \\
           \and
           Jingyi Jessica Li* (corresponding author) \at
           Department of Statistics, University of California, Los Angeles, CA 90095-1554, USA \\
           Department of Human Genetics, University of California, Los Angeles, CA 90095-7088, USA\\
           \email{jli@stat.ucla.edu}  
}

\date{Received: 3 January 2016 / Accepted: 29 July 2016}

\maketitle

\begin{abstract}
Comparative transcriptomics has gained increasing popularity in genomic research thanks to the development of high-throughput technologies including microarray and next-generation RNA sequencing that have generated numerous transcriptomic data. An important question is to understand the conservation and differentiation of biological processes in different species. We propose a testing-based method TROM (Transcriptome Overlap Measure) for comparing transcriptomes within or between different species, and provide a different perspective to interpret transcriptomic similarity in contrast to traditional correlation analyses. Specifically, the TROM method focuses on identifying associated genes that capture molecular characteristics of biological samples, and subsequently comparing the biological samples by testing the overlap of their associated genes. We use simulation and real data studies to demonstrate that TROM is more powerful in identifying similar transcriptomes and more robust to stochastic gene expression noise than Pearson and Spearman correlations. We apply TROM to compare the developmental stages of six \emph{Drosophila} species, \emph{C. elegans}, \noteW{\emph{S. purpuratus}, \emph{D. rerio}} and mouse liver, and find interesting correspondence patterns that imply conserved gene expression programs in \JJL{the} development \JJL{of these species}. The TROM method is available as an R package on CRAN (\href{https://cran.r-project.org/package=TROM}{https://cran.r-project.org/package=TROM}) with manuals and source codes available at \href{http://www.stat.ucla.edu/~jingyi.li/software-and-data/trom.html}{http://www.stat.ucla.edu/~jingyi.li/software-and-data/trom.html}.
\end{abstract}

\section{Introduction}

Comparative genomics is an important field \JJL{that addresses} evolutionary questions and \JJL{studies} developmental processes across distant species \cite{Pantalacci2015}. Studying transcriptomes is essential for understanding functions of genomic regions and interpreting regulatory relationships of multiple genomic elements \cite{Wang2009}. Comparing transcriptomes of the same species can reveal molecular mechanisms behind the \noteJJL{occurrence} and progression of important biological processes, such as organism development and stem cell differentiation \cite{Shen2012,Labbe2012}. Comparing transcriptomes of different species can help understand the conservation and differentiation of \noteJJL{these} molecular mechanisms in evolution \cite{Li2014}. High-throughput technologies have generated large amounts of publicly available transcriptomic data, creating an unprecedented opportunity for comparing multi-species transcriptomes under various biological conditions.

Finding the transcriptomic similarity and disparity of biological samples is a key step to understand the underlying molecular mechanisms common or unique to them. It is desirable to have a transcriptomic similarity measure that can lead to a clear correspondence pattern of biological samples from the same or different species. Correlation analysis is a classical approach for comparing transcriptomes based on gene expression data. Commonly used measures are Pearson and Spearman correlation coefficients, both of which have played important roles in biological discoveries \cite{Arbeitman2002,Spencer2011,Necsulea2014}. However, in most scenarios neither of them can produce a clear correspondence pattern among biological samples. The main reason is the existence of many housekeeping genes, which would inflate \JJL{correlation} coefficients. \noteJJL{Moreover, correlation measures rely heavily on the accuracy of gene expression data and are \JJL{susceptible} to the \JJL{low signal-to-noise ratios of} lowly expressed genes.} Therefore, it is often difficult to use correlation analysis to find a clear correspondence pattern \noteJJL{of transcriptomes}.

Here we introduce a new \noteJJL{testing-based} measure---transcriptome overlap measure \noteJJL{(TROM)}---to find correspondence \noteJJL{of transcriptomes} in the same or different species. The measure is based on testing the overlap of ``associated genes,'' which represent transcriptomic characteristics of biological samples.
For the purpose of discovering sparse sample relationships, we define a \textit{sample correspondence map} as the binarized mapping pattern resulted from a sample similarity matrix: a none-zero value means that two samples are \JL{\textit{mapped}} to each other, while a zero value means that two samples are \JL{\textit{unmapped}}.
We show that compared to Pearson and Spearman correlations, TROM has better power to detect transcriptome correspondence in simulations and leads to clearer correspondence maps of developmental stages within and between multiple species in real data studies. TROM also provides a systematic approach for selecting associated genes of every biological sample. We show that these associated genes can well capture transcriptomic characteristics and help construct developmental trees in multiple species. In addition, we demonstrate that TROM is robust to data normalization and high-throughput platform difference.

In Section \ref{sec:method}, we describe the TROM method including the identification of associated genes, the calculation of TROM scores, and the selection of a threshold parameter. In Section \ref{sec:results}, we present real data applications of TROM to large-scale transcriptomic data sets, power analysis of TROM versus Pearson and Spearman correlations, demonstration of the robustness of TROM to data normalization and platform difference, and bioinformatic analyses of the TROM results.

\section{Method}\label{sec:method}
\subsection{\JJL{Associated genes and TROM scores}}
Our method focuses on selecting associated genes to perform a gene set overlap test \cite{Li2014}, which will lead to TROM scores that can be used to compare biological samples. We define \JJL{\textit{associated genes} of a sample} using the following criterion: the genes that have $z$-scores (normalized expression levels across samples) $\geq z$ in the sample, where $z$ is a threshold \noteJJLnew{that can be selected in a systematic approach (please see Section \ref{sec:thre}) or set by users}. \JJL{Based on this definition, associated genes of a sample are those with higher expression in the sample compared to a few other samples. In other words, associated genes are highly expressed in the sample of interest but not always highly expressed in all samples, and they are a superset of sample specific genes. Hence, associated genes capture gene expression characteristics of a sample, and these characteristics are either specific to the sample or shared by a few other samples but not all samples. Associated genes provide a basis for comparing biological samples.} We compare two biological samples by statistically testing the dependence of their associated genes: to compare two samples of the same species, we calculate the significance of the number of their overlapping associated genes (resulting in a within-species TROM score); to compare two samples of different species, we calculate the significance of the number of orthologous gene pairs in their associated genes (resulting in a between-species TROM score).

We consider the two sample-associated gene sets as two samples drawn from the gene population. 
In the within-species scenario, we denote the number of biological samples \JJL{of a} given species as $m$\noteJJLnew{, and use} $X_i$ and $X_j$ ($i,j=1,2,\dots,m$) to denote the associated genes of samples \JJL{$i$ and $j$} to be compared. The gene population consists of all genes of the given species\noteJJLnew{, and the size of the gene population} is denoted as $N$. 
Then to test for the null hypothesis that $X_i$ and $X_j$ are two independent samples drawn from the gene population versus the alternative hypothesis that $X_i$ and $X_j$ are dependent samples, the $p$-value for within-species comparison between \JJL{samples $i$ and $j$} is calculated as
\begin{equation}\label{eq1}
p\text{-value} = \sum_{k = |X_i\cap X_j|}^{\min(|X_i|,|X_j|)} \frac{{N \choose k}{N-k \choose |X_i|-k} {N-|X_i| \choose |X_j|-k}}{{N \choose |X_i|}{N \choose |X_j|}}.
\end{equation}

In the between-species scenario, we denote the numbers of biological samples from \JJL{species 1 and 2} as \noteJJLnew{$m_1$} and $m_2$\noteJJLnew{. The gene population consists of all orthologous gene pairs between the two species, and the number of pairs is denoted as $N$.
\noteW{
The ortholog pairs can be represented as a two-column table with $N$ rows. 
}
 We use $X_i$ ($i=1,2,\dots,m_1$) and $Y_j$ ($j=1,2,\dots,m_2$) to denote the orthologous gene pairs (i.e., rows in the table) that overlap with the associated genes of sample $i$ in species 1 and sample $j$ in species 2, respectively. 
 In other words, $X_i$ (or $Y_j$) represents the orthologous gene pairs that contain the associated genes in sample $i$ of species 1 (or sample $j$ of species 2).
  }  Then to test for the null hypothesis that $X_i$ and $Y_j$ are two independent samples drawn from the population of orthologous gene pairs versus the alternative hypothesis that $X_i$ and $Y_j$ are dependent samples, the $p$-value for between-species comparison \noteJJLnew{of the two samples} is calculated as
\begin{equation}\label{eq2}
p\text{-value} = \sum_{k = |X_i \cap Y_j|}^{\min(|X_i|,|Y_j|)} \frac{{N \choose k}{N-k \choose |X_i|-k} {N-|X_i| \choose |Y_j|-k}}{{N \choose |X_i|}{N \choose |Y_j|}}.
\end{equation}
Then we define the within-species or between-species TROM score as
\begin{equation}\label{eq3}
\text{TROM score} = -\log_{10}(\text{Bonferroni-corrected}\ p\text{-value}),
\end{equation}
\JJL{which describes transcriptome similarity of two biological samples. A larger TROM score represents greater similarity.}

\subsection{Selection of $z$-score threshold}\label{sec:thre}

\begin{figure}[!tpb]
\centering
\includegraphics[width=0.9\linewidth]{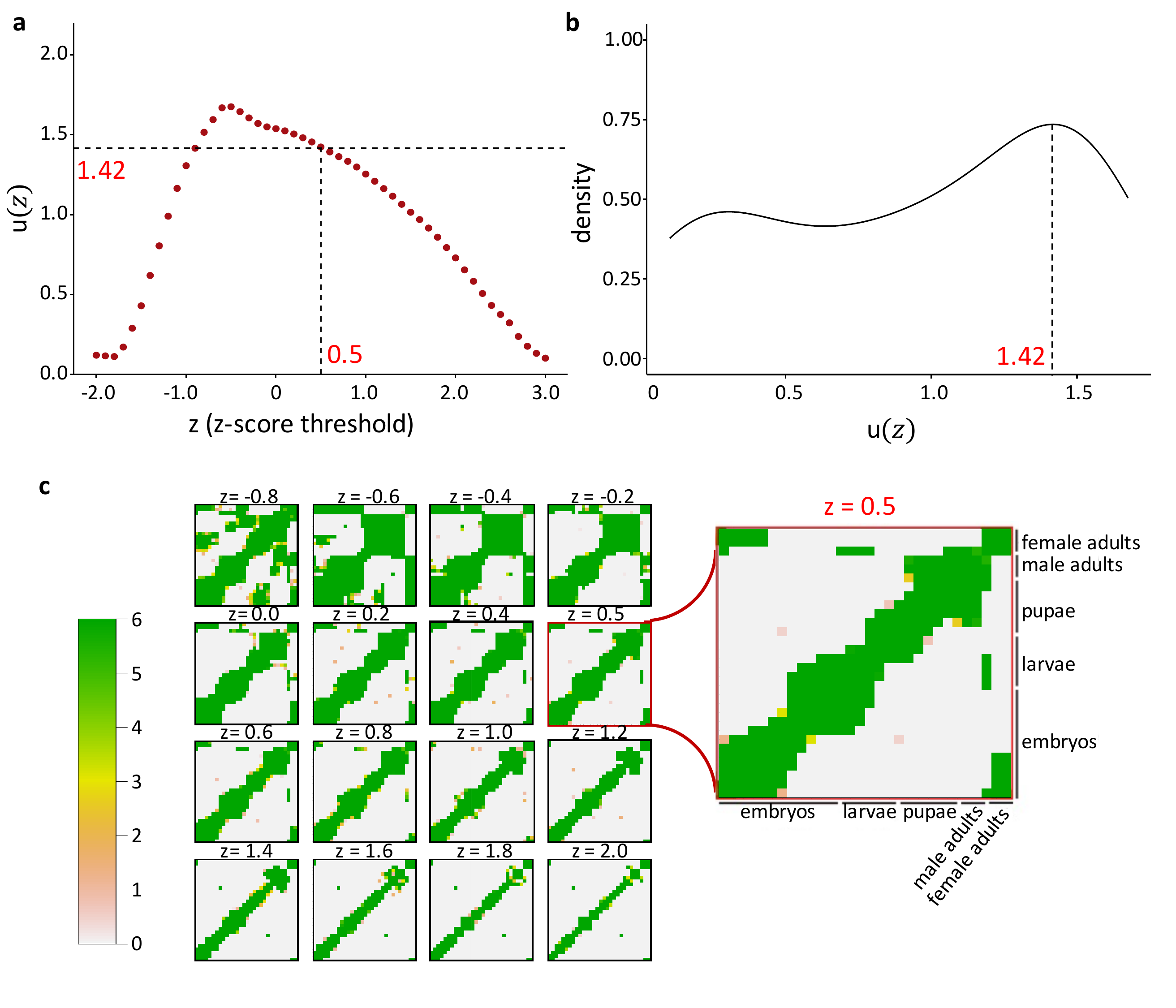}
\caption{Selection of $z$-score threshold for comparing \emph{D. melanogaster} (fly) developmental stages. \textbf{a}: Values of $u(z)$ at different $z$-score thresholds. The horizontal dashed line marks the $\text{mode}(u)$ shown in \textbf{b}, 1.42, which corresponds to $z = 0.5$, the selected $z$-score threshold. \textbf{b}: The density plot of $u(z)$ with Gaussian kernel and banwidth = $0.22$, based on the $u(z)$ values shwon in \textbf{a}. \textbf{c}: Changes of TROM correspondence maps (for 30 fly stages) as the $z$-score thresholds (marked on top of each correspondence map) change. The inset heatmap shows the correspondence map of the chosen threshold $z^* = 0.5$. In each heatmap, both columns and rows represent fly's $30$ developmental stages, and darker colors represent larger TROM scores.
}
\label{zthre}
\end{figure}

The selection of the $z$-score threshold $z$ will directly influence the \JJL{sensitivity and} specificity of sample-associated genes and thus \JJL{affect} the resulting TROM scores. If \noteJJLnew{$z$ is too small}, a large \noteJJLnew{number of associated genes will be selected for every sample and more associated genes will be shared by \JJL{different} samples, \JJL{and thus it becomes} difficult to distinguish different biological samples. If $z$ is too large, only a small number of associated genes will} be identified for each sample and potentially informative genes could be filtered out, \JJL{and thus no similarity of biological samples will be captured by TROM}. 
Although the selection of \noteJJLnew{$z$} is \JJL{ultimately subject to} \noteJJLnew{users'} preference for the resulting correspondence maps \noteJJLnew{(a larger $z$ for a sparser map or a smaller $z$ for a denser map)}, we propose an objective approach to choose \JJL{an} appropriate threshold when no prior knowledge is available. Our approach aims at \noteJJLnew{balancing} two goals: \noteJJLnew{(1)} the threshold should \JJL{help} minimize \JJL{noisy correspondence of biological samples} and thus leads to \JJL{a} sparse correspondence \JJL{map}; \noteJJLnew{(2)} the threshold should help preserve \JJL{strong correspondence of samples} and thus leads to \JJL{a} stable correspondence \JJL{map}.

We use the mean of TROM scores \noteJJLnew{of all pairwise \JJL{comparisons of biological samples} in the correspondence map} as the \noteJJLnew{objective function, which} is defined as
\begin{equation}\label{eq4}
\JJL{ 
u(z)=\log_{10}\left(\frac{\sum^m_{i=1}\sum^m_{j=1, j\neq i}a_{ij}(z)}{m^2-m} + 1\right)
}
\end{equation}
where $m$ is the number of biological samples, \JJL{$A(z)=\left(a_{ij}(z)\right)_{m\times m}$} is the TROM \JJL{score} matrix based on threshold $z$. 
We select the desirable threshold $z^*$ by the following approach. Considering our goal \JJL{(2)}, we would like $u(z)$ 
%
to be stable for $z$ values near \JJL{$z^*$}. Since similar $u(z)$ values would lead to a peak in the density of \noteW{$u(z)$, denoted as $f(u)$,} we consider the $z$ values corresponding to the peak, that is, $\{z: u(z) =\text{mode}(u)\}$, where mode$(u) = \arg\max_u f(u)$ (i.e., the $u$ value that maximizes the density of $f(u)$ for $u = u(z)$ with $z \in [-2,3]$). Also considering our goal (1), we would like to select $z^*$ as the largest $z$ value that \JJL{leads to} the stable region of $u(z)$. Hence, we find $z^*$ as
\begin{equation}
z^* = \sup\left\{z: u(z) = \text{mode}(u)\right\},   
\end{equation}
where $u = u(z) \text{ for } z \in [-2,3]$. If users desire a sparser correspondence map, we suggest an alternative approach to finding the $z$-score threshold as $z^* = \sup \{ z: u(z) = \text{mode}(u) + \text{sd}(u) \}$, where $\text{sd}(u)$ stands for the standard deviation of the $u(z)$ values. According to Lemma \ref{lemma1} and also our empirical observation, $[-2,3]$ is a large enough region to capture the peak with low computational intensity, as the $u(z)$ values are close to $0$ outside of this region.

\noteW{
As shown in Lemma ~\ref{lemma1}, an important feature of $u(z)$ is that it approaches $0$ when the absolute value of $z$ is large. This is because the entire gene population will be selected as associated genes when the threshold $z$ is small enough while no genes will be selected when $z$ is large enough. In both extreme cases, the resulting TROM score is $0$ for any pair of samples. Because of this feature and the non-negativity of $u(z)$, $u(z)$ must have a maximum at a certain value of $z$. The observed unimodal shape is a typical feature of $u(z)$ for the various species we have investigated. }

\begin{lemma}\label{lemma1}
	$u(z) \rightarrow 0$ as $|z| \rightarrow \infty$.
\end{lemma}
\begin{proof}
\small
Because of the criterion of selecting associated genes: $z$-scores $\ge z$, for within-species comparison between samples $i$ and $j$, whose sets of associated genes are denoted as $X_i$ and $X_j$, we have
\begin{itemize}
\item as $z \rightarrow -\infty$, $|X_i| \rightarrow N$, $|X_j| \rightarrow N$, and $|X_i \cap X_j| \rightarrow N$, where $N$ is the number of all genes of the species;
\item as $z \rightarrow \infty$, $|X_i| \rightarrow 0$, $|X_j| \rightarrow 0$, and $|X_i \cap X_j| \rightarrow 0$.
\end{itemize}
Given the $p$-value formula (Equation~(\ref{eq1})) of the within-species overlap test in TROM, we have
\begin{itemize}
\item as  $|X_i| \rightarrow N$, $|X_j| \rightarrow N$, and $|X_i \cap X_j| \rightarrow N$, $p$-value $\rightarrow 1$;
\item as $|X_i| \rightarrow 0$, $|X_j| \rightarrow 0$, and $|X_i \cap X_j| \rightarrow 0$, $p$-value $\rightarrow 1$.
\end{itemize}

For between-species comparison between samples $i$ from species 1 and sample $j$ from species 2, whose associated genes correspond to ortholog pairs denoted as $X_i$ and $Y_j$, and between $X_i$ and $Y_j$ there are $m_0$ ortholog pairs, we have
\begin{itemize}
\item as $z \rightarrow -\infty$, $|X_i| \rightarrow N$, $|Y_j| \rightarrow N$, and $m_0 \rightarrow N$, where $N$ is the total number of ortholog pairs between the two species;
\item as $z \rightarrow \infty$, $|X_i| \rightarrow 0$, $|Y_j| \rightarrow 0$, and $m_0 \rightarrow 0$.
\end{itemize}
Given the $p$-value formula (Equation~(\ref{eq2})) of the between-species overlap test in TROM, we have
\begin{itemize}
\item as $|X_i| \rightarrow N$, $|Y_j| \rightarrow N$, and $m_0 \rightarrow N$, $p$-value $\rightarrow 1$;
\item as $|X_i| \rightarrow 0$, $|Y_j| \rightarrow 0$, and $m_0 \rightarrow 0$, $p$-value $\rightarrow 1$.
\end{itemize}

So for both within-species and between-species comparisons, we have TROM score $a_{ij}(z) \rightarrow 0$ as $|z| \rightarrow \infty$ given Equation~(\ref{eq3}).

Hence, given the definition of $u(z)$ in Equation~(\ref{eq4}), we have $u(z) \rightarrow 0$ as $|z| \rightarrow \infty$.
\end{proof}

Using this proposed approach, we can easily select a $z$-score threshold for a specified species given its gene expression data. 
We demonstrate how this approach can select an appropriate threshold for comparing \emph{D. melanogaster} developmental stages by applying it to the RNA-seq data of $m=30$ stages. We consider candidate thresholds in the range of $z \in [-2, 3]$ and calculate TROM matrices for all the candidate values in this range with a step size of $0.1$. The corresponding $u(z)$ is plotted in Figure \ref{zthre}a.

From the density of $u(z)$ (see Figure \ref{zthre}b), we determine that the mode of $u(z)$ is 1.42. By finding the maximum $z$ value such that $u(z)=1.42$, our approach selects $z^* = 0.5$.
Figure \ref{zthre}c shows how different $z$-score thresholds can influence the patterns of correspondence maps. When the threshold is too low (e.g., $-0.4$), many stage pairs are mapped to each other, \JJL{providing} vague information on the relationships of different stages. On the other hand, when the threshold is too high (e.g., $2.0$), \JJL{so} much information is filtered out that most stages are only mapped to themselves, and important correspondence such as the \JJL{similarity} between fly early embryos and female adults is missing \cite{Li2014}. \JJL{Unlike the two extremes}, our selected threshold $0.5$ \JJL{reveals} important correspondence patterns and \JJL{meanwhile yields} a clean correspondence map. 


\section{Results}\label{sec:results}

\subsection {Application of TROM to finding correspondence of developmental stages of multiple species}\label{sec:TROM scores}

\JJL{We first demonstrate the use and the performance of TROM in comparative transcriptomics.} We apply TROM to find correspondence patterns of developmental stages of six \textit{Drosophila} (fly) species, \textit{C. elegans} (worm), \textit{S. purpuratus} (sea urchin), \textit{D. rerio} (zebrafish) and \JJL{mouse} liver tissues. The goal is to find similarity of developmental stages within and between species in terms of gene expression dynamics. We use multiple datasets including RNA-seq data of $30$ \textit{D. melanogaster} developmental stages with expression estimates of $15,095$ genes, RNA-seq data of $35$ \emph{C. elegans} stages with $31,622$ genes \cite{gerstein2014comparative,Li2014}, RNA-seq data of $10$ sea urchin stages with $21,090$ genes \cite{tu2014quantitative}, microarray data of six fly species: \emph{D. melanogaster, D. simulans, D. ananassae, D. persimilis, D. pseudoobscura} and \emph{D. virilis} with $9$ to $13$ embryonic stages and $3,663$ genes \cite{Arbeitman2002}, microarray data of mouse liver development with $14$ stages and $45,101$ genes \cite{li2009multi} 
and microarray data of \textit{D. rerio} with $61$ stages and \noteW{$18,259$} genes \cite{domazet2010phylogenetically}. To implement TROM on these gene expression datasets, we select $z$-score thresholds based on the alternative approach described in Section \ref{sec:thre}, and the selected thresholds for various species are summarized in Appendix Table \ref{tab:z_thre} and used throughout this paper unless otherwise specified. A detailed description of these datasets is given in Appendix Table \ref{stage_label}.

In the comparison of developmental stages within each species, the TROM method finds \JJL{block} diagonal correspondence patterns as expected. That is, in every species, adjacent developmental stages close to each other in the time order have high TROM scores. We illustrate the correspondence maps of developmental stages of mouse liver (Figure~\ref{TROM_map}a), sea urchin (Figure~\ref{TROM_map}b) 
and the six \textit{Drosophila} species (Appendix Figure \ref{TROM_scores}). These results provide strong support to the efficacy and validity of TROM in finding transcriptomic similarity of biological samples, in addition to our previous results on \JJL{the} correspondence of \textit{D. melanogaster} and \textit{C. elegans} stages based on RNA-seq data \cite{Li2014}, \JJL{to which} we applied the preliminary idea of TROM.

\begin{figure}[!tb]
\centering
\includegraphics[width=0.75\linewidth]{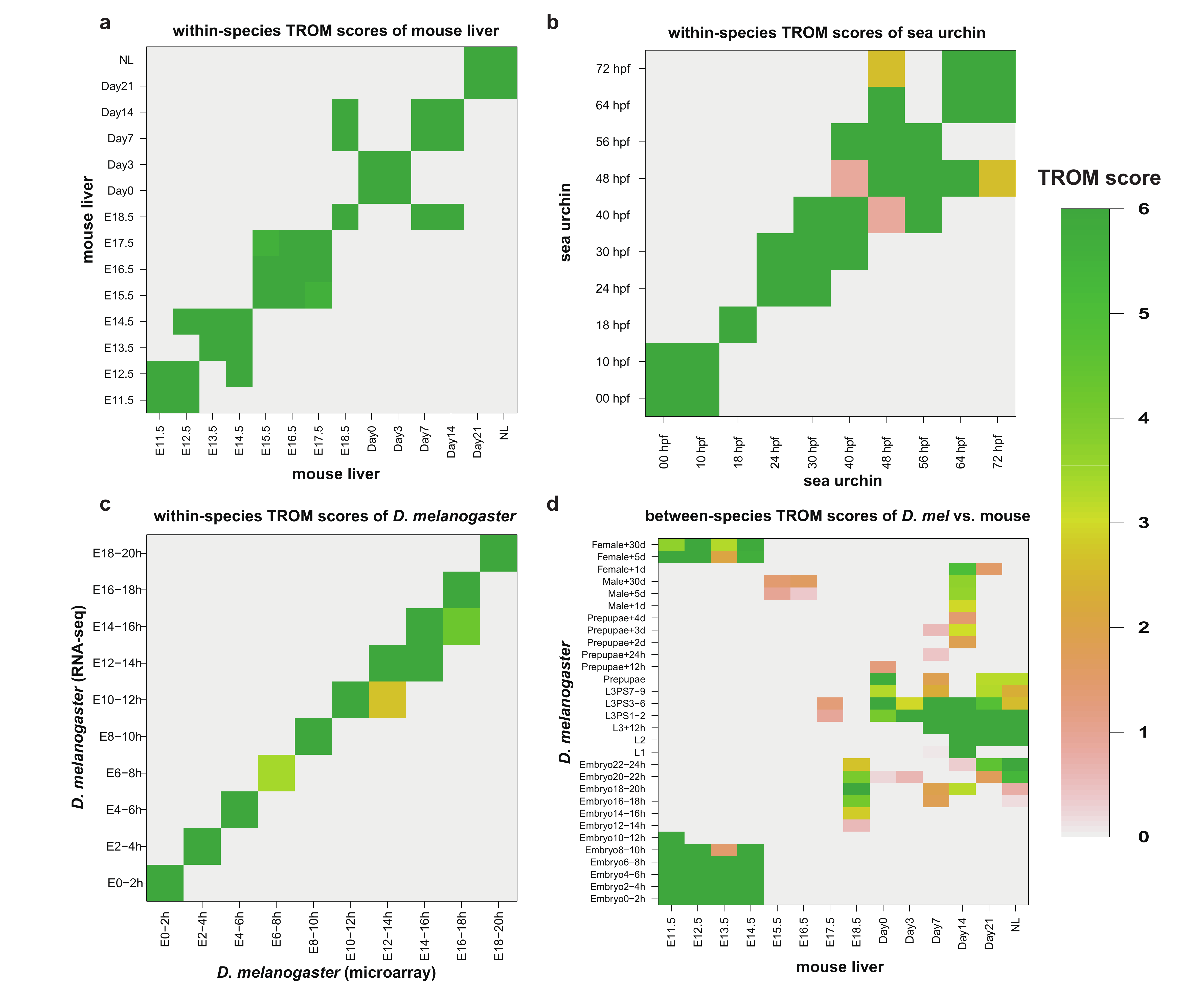}
\caption{ Within-species and between-species correspondence maps of TROM scores. For better illustration, TROM scores are saturated at $6$: all the scores larger than $6$ are set to $6$. \textbf{a}: pairwise within-species TROM scores calculated for the $14$ stages of mouse liver; \textbf{b}: pairwise within-species TROM scores calculated for the $10$ stages of sea urchin; \textbf{c}: pairwise within-species TROM scores calculated for the $10$ stages of \emph{D. melanogaster}. The column stages \JJL{are from} the microarray data, \JJL{and} the row stages \JJL{are from} the RNA-seq data; \textbf{d}: pairwise between-species TROM scores of \emph{D. melanogaster} vs. mouse liver. The columns represent the $14$ mouse stages, \JJL{and} the rows represent the $30$ fly stages.}
\label{TROM_map}
\end{figure}

We also apply TROM to compare the developmental stages of two different species. \noteW{We use ortholog information downloaded from Ensembl \cite{cunningham2015ensembl} in the comparison}. Since fly, worm and mouse are vastly distant from each other in evolution, any correspondence between their developmental stages revealed by TROM will be interesting and may imply conserved developmental programs. Between \textit{D. melanogaster} life cycle and mouse liver development (Figure~\ref{TROM_map}d), TROM finds unknown correspondence between fly early embryos and \JJL{mouse embryo liver tissues}, and between fly female adults and \JJL{mouse embryo liver tissues}. A \JJL{main} reason for the latter correspondence is the transcriptomic similarity of fly early embryos and female adults due to the expression of maternal effect genes \cite{Li2014}. Additionally, there is some irregular correspondence between fly larvae and liver \JJL{tissues of} born mice. We can see a clear separation of \JJL{the liver tissues of} mouse embryos and born mice, \JJL{and their} corresponding fly stages also exhibit a separation of embryos and female adults from other stages. These results \JJL{indicate} that even for vastly different species such as fly and mouse, there is good conservation in their embryonic development. Similarly between the six \textit{Drosophila} species' embryonic development and mouse liver development, we also see good correspondence of fly early embryos and \JJL{mouse embryo liver tissues}, and correspondence between fly late embryos and \JJL{mouse adult liver tissues} (Appendix Figure \ref{TROM_scores}). Moreover, \JJL{mouse embryo liver tissues are} observed to correspond well with worm embryos, \JJL{and this is consistent with the observed correspondence between fly embryos and worm embryos (Appendix Figure \ref{TROM_scores})}.  These consistent correspondence patterns together validate the efficacy of the TROM approach.

Between the six \textit{Drosophila} species, since they are known to have similar developmental programs \cite{Arbeitman2002}, comparisons of their developmental stages \JJL{resemble} within-species comparisons, and \JJL{block} diagonal correspondence patterns \JJL{are} expected. Our results confirm this\noteW{: d}iagonal patterns are observed between the developmental stages of every two fly species (Appendix Figure \ref{TROM_scores}). These results again demonstrate the validity of TROM.

\subsection{Comparison of TROM and Pearson/Spearman correlation measures}\label{sec:power_simu}
We next describe the scenarios where TROM serves as a better similarity measure than Pearson/Spearman correlation measures \JL{in differentiating the stage pairs, which exhibit high dependence in highly expressed genes, from other stage pairs.} A key difference between our TROM method and the Pearson/Spearman correlation analysis is that TROM divides genes into two sets (associated genes and non-associated genes) for every sample based on gene expression dynamics across all samples. After the division, calculation of TROM scores does not rely on actual gene expression measurements. \JJL{Henceforth, TROM defines \JL{sample similarity} based on the overlap of their associated genes.} In contrast to TROM, Pearson and Spearman correlations are calculated based on actual expression measurements of the same set of genes in two samples. Hence, they are more sensitive to expression fluctuations of lowly expressed genes due to measurement errors, and their values can be driven high by \JL{the genes (e.g., housekeeping genes)} that have approximately constant expression across samples and carry little information on sample characteristics. For our goal of constructing a \JL{sparse sample correspondence map based on gene expression}, Pearson and Spearman correlation measures are \JL{often unsatisfactory, as they give} rise to noisy correspondence maps (Appendix Figures \ref{outline} and \ref{cor_associated}).

To demonstrate the power of TROM in detecting the correspondence of biological samples that share transcriptomic characteristics embedded in \noteJJL{highly expressed genes}, we conduct a simulation study to compare TROM \JL{with Pearson and Spearman correlation measures. 
Specifically, we consider their values as classification scores to differentiate the sample pairs with strong dependence in highly expressed genes from the rest sample pairs.} We evaluate their performance in terms of classification accuracy.

Suppose a species of interest has a total number of $N$ genes and $m$ samples. For the observed data, let $\bm X_j = (X_{1j}, \ldots, \allowbreak X_{Nj})^T$ denote the expression vector of the $N$ genes in sample $j$. For the underlying (hidden) sample similarity, we use a state matrix $\bm E_{m\times m}$ to denote the pairwise relationships between the $m$ samples. That is, if samples $i$ and $j$ have high dependence in their associated genes, $E_{ij} = 1$; otherwise $E_{ij} = 0$. We consider how to predict $E_{ij}$ for every pair $1 \le i \neq j \le m$ from gene expression matrix as a classification problem. We would like to compare the three measures in this setting and evaluate their performance as classification scores using precision recall curves, receiver operating characteristic (ROC) curves, and Neyman-Pearson
ROC curves \cite{nproc}.

In this simulation, we define the state matrix  $\bm E_{m\times m}$ based on a correlation matrix of associated genes. Specifically, \JL{in the example of comparing developmental stages,} we assume a Toeplitz-type correlation matrix $\bm\Sigma$ where $\Sigma_{ij} = \rho^{|i-j|}\ (i,j=1,2,\dots,m;~ \rho\in [0,1])$, \JL{which is reasonable as it assigns a higher correlation to more adjacent stage pairs. To reduce arbitrariness in defining $\bm E$ based on $\bm\Sigma$, we vary a threshold $c \in (0,1)$ and define $\bm E$ as
\begin{eqnarray}
E_{ij} = 
\left\{
\begin{array}{ll}
1 & \text{if}\ \Sigma_{ij} > c\\
0 & \text{if}\ \Sigma_{ij}\leq c\\
\end{array}\right..
\label{defineE}
\end{eqnarray}
We would like to track how the classification accuracy of the three measures changes as the parameter $c$ changes. 
}

We use the following generative model to simulate gene expression matrices. We let $\bm I_{N\times m}$ be \JL{an indicator matrix, with $I_{ij} = 1$ if gene $i$ is an associated gene of sample $j$ and $I_{ij} = 0$ otherwise}. Given the correlation matrix $\bm\Sigma$, we assume that the $i$th row $\bm I_i \in \{0,1\}^m$ is a binary vector randomly sampled from a multivariate Bernoulli distribution with expectation $q\times \bm 1_{m\times 1}$ (\JL{$q \in (0,1)$} inferred from real data) and correlation matrix $\bm\Sigma$. Given \JL{the associated-gene indicator matrix} $\bm I_{N\times m}$, we generate a gene expression matrix in a data-driven approach, because gene expression values in real data contain noises and cannot be easily described by any common probability distributions. \JL{We first scale a real gene expression matrix $\bm Y_{N\times m}$ by dividing each of its rows by the row maximal values, denoted by $\bm Y^{\text{scale}}$.} Then for each gene $i=1,2,\dots,N$, we locate its closest counterpart in real data by searching for gene $i'$ in $\bm Y^{\text{scale}}$ such that \JL{the $i$th row} $\bm Y_{i'}^{\text{scale}}$
and $\bm I_i$ has the minimal Euclidean distance. Given $Y_{i'}$, the $i'$th row of $\bm Y$, we define sets $A_{i'} = \left\{ Y_{i'j}: \text{gene $i'$ is an associated gene in sample $j$, } j=1,\ldots,m \right\}$ and $A_{i'}^c = \left\{ Y_{i'j}: \text{gene $i'$ is not an associated gene in sample $j$, } j=1,\ldots,m \right\}$ to collect the expression values of gene $i'$ when it is identified as associated or not associated with real-data samples, based on a pre-determined $z$-score threshold $z^*$. Finally, we create a gene expression matrix $\bm{X}_{N \times m}$ as follows: for gene $i$ in sample $j$, if $I_{ij} = 1$, we randomly sample the value of $X_{ij}$ from $A_{i'}$; if $I_{ij} = 0$, we randomly sample the value of $X_{ij}$ from $A_{i'}^c$.

Using this generative model, we simulate $K = 200$ gene expression matrices of the same species. We denote the matrices as $\bm X^{(k)}$, $k=1,\ldots,K$. Then we calculate the \JL{similarity score} matrices based on the three similarity measures. \JL{For TROM, to} determine the associated genes and non-associated genes of each sample, we calculate the \JL{$z$-score} threshold based on $\bm X^{(k)}$ using the method introduced in Section \ref{sec:thre}. The resulting TROM score \JL{matrix is} denoted as $\bm T^{(k)}$. The Pearson and Spearman correlation matrices are denoted as $\bm P^{(k)}$ and $\bm S^{(k)}$, respectively. Please note that $\bm{T}^{(k)}$, $\bm{P}^{(k)}$ and $\bm{S}^{(k)}$ are all $m \times m$ matrices, with the same dimensions as $\bm{E}$.

\JL{To perform classification based on the score matrices of the three measures, we apply multiple cutoffs to the matrices and calculate the resulting precision and recall rates. For example, if we use $c_T$ as the cutoff for TROM scores, for $k= 1,2,\dots,K$ we have predicted class labels}
\begin{eqnarray*}
\hat E^{(k)}_{ij} = 
\left\{
\begin{array}{ll}
1 & \text{if}\ T^{(k)}_{ij} > c_T\\
0 & \text{if}\ T^{(k)}_{ij}\leq c_T\\
\end{array}\right..
\end{eqnarray*}
The precision and recall \JL{rates} of TROM in \JL{the $k$th run are then} calculated as
\begin{eqnarray*}
\text{precision}^{(k)} &=& \frac{ \underset{i\neq j}{\sum\sum} \hat E^{(k)}_{ij}E_{ij}}{\underset{i\neq j}{\sum\sum}  \hat E^{(k)}_{ij}}~,\\
\text{recall}^{(k)} &=& \frac{ \underset{i\neq j}{\sum\sum} \hat E^{(k)}_{ij}E_{ij}}{\underset{i\neq j}{\sum\sum}  E_{ij}}~.\\
\end{eqnarray*}
Similarly, we can calculate the precision and recall \JL{rates} of Pearson/Spearman correlation by applying \JL{varying} cutoffs on $\bm P^{(k)}$ and  $\bm S^{(k)}$ respectively.

\begin{figure}[!tpb]
\centering
\includegraphics[width=\linewidth]{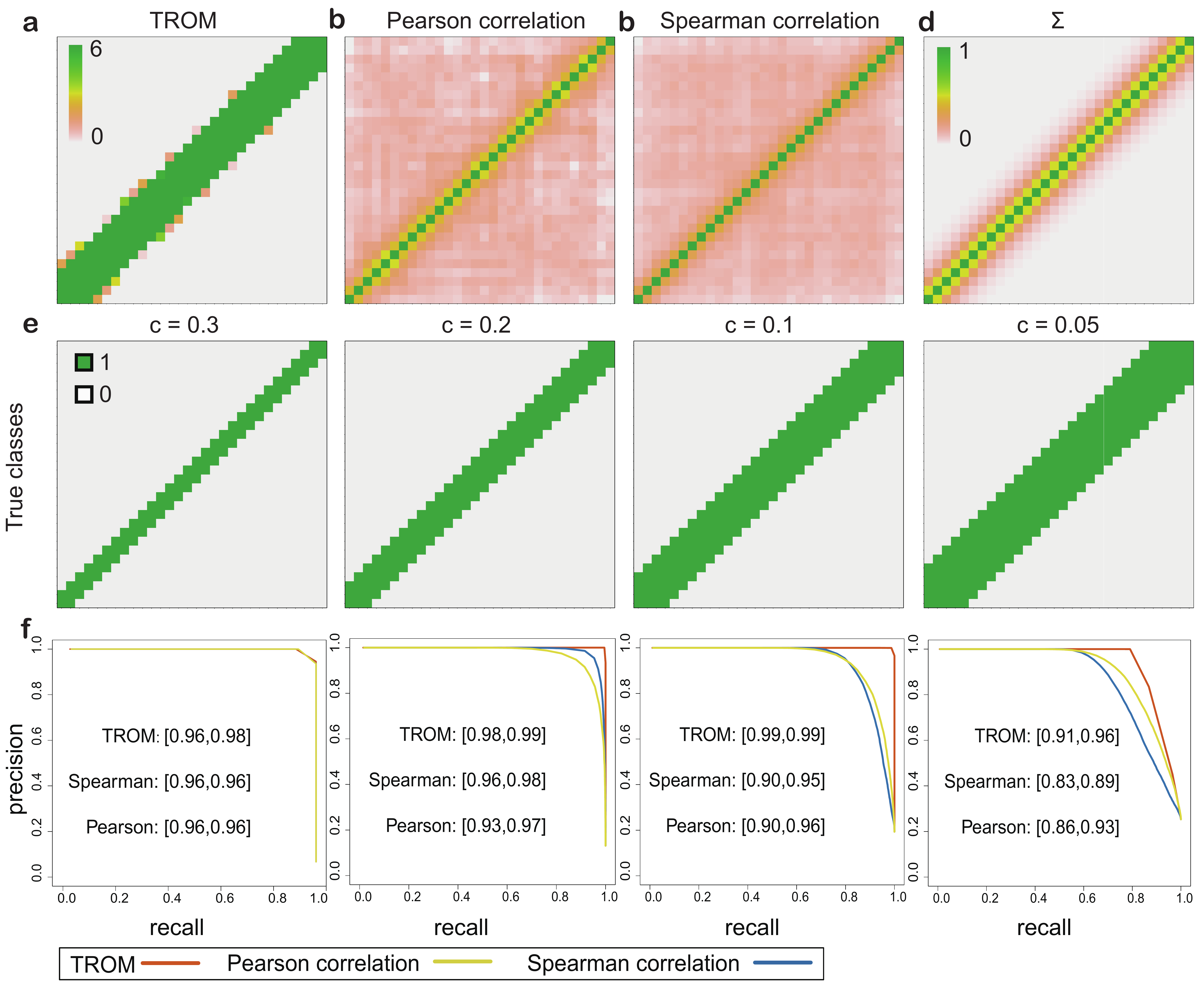}
\caption{Comparison of TROM and Pearson/Spearman correlation on simulated \emph{D. melanogaster} (fly) data. \JL{\textbf{a-c}: The correspondence maps produced by TROM (\textbf{a}), Pearson correlation (\textbf{b}) and Spearman correlation (\textbf{c}) on a randomly selected gene expression matrix (among the $K = 200$ matrices). \textbf{d}: The correlation matrix $\bm \Sigma$ that defines the dependence of associated genes between samples. \textbf{e}: The true sample relationships ($1$: high dependence in associated genes; $0$: otherwise)  defined as in Equation \ref{defineE} for varying $c$. \textbf{f}: The mean precision-recall curves on the $200$ gene expression matrices, given the true labels in \textbf{e}. The 95\% confidence intervals of each measure's area under the curve (AUC) are marked next to the curves.}}
\label{simu_fly}
\end{figure}

We carry out \JL{this} simulation study \JL{in the context of} \emph{D. melonagaster} (fly) and \emph{C. elegans} (worm). For fly, we \JL{have} $N = 10,000$, $m = 30$, $q = 0.15$, $z^* = 0.5$; for worm, we \JL{have} $N = 10,000$, $m = 35$, $q = 0.2$, $z^* = 0.6$. In both cases, we set $\rho = 0.5$ and let $c$ take four different values: 0.3, 0.2, 0.1, 0.05. The real data used to generate the simulated gene expression \JL{matrices} are processed from modENCODE RNA-seq data of $30$ fly developmental stages and $35$ worm developmental stages \cite{gerstein2014comparative,Li2014}. The precision-recall curves of the three measures are illustrated in Figures \ref{simu_fly} and \ref{simu_worm}. \JL{In both cases, we see that TROM produces clearer sparse patterns of sample similarity (Figure \ref{simu_fly}a vs. b-c and Figure \ref{simu_worm}a vs. b-c), and for predicting stage-pair labels defined by different threshold $c$ values, TROM always has the largest area under the precision-recall curves (Figures \ref{simu_fly}e-f and \ref{simu_worm}e-f, in terms of both the mean area and the 95\% confidence intervals from the $K=200$ simulation runs).} We also calculate Receiver Operating Characteristic (ROC) and the Neyman-Pearson Receiver Operating Characteristic (NP-ROC \cite{nproc}) curves of the three measures in each case (see Appendix Figure \ref{simu_roc}), and \JL{TROM still has the best classification accuracy}.

\JL{In this classification setting, TROM scores, Pearson correlations and Spearman correlations are essentially three ways of transforming a gene expression matrix into features of sample pairs. The above simulation results suggest that TROM scores serve as better features for this task, that is, to capture the sparse similarity relationships of samples. The main reason is that TROM scores are based on gene expression levels of all samples, while Pearson and Spearman correlations only capture the similarity of gene expression profiles for every pair of samples.}

\begin{figure}[!tpb]
\centering
\includegraphics[width=\linewidth]{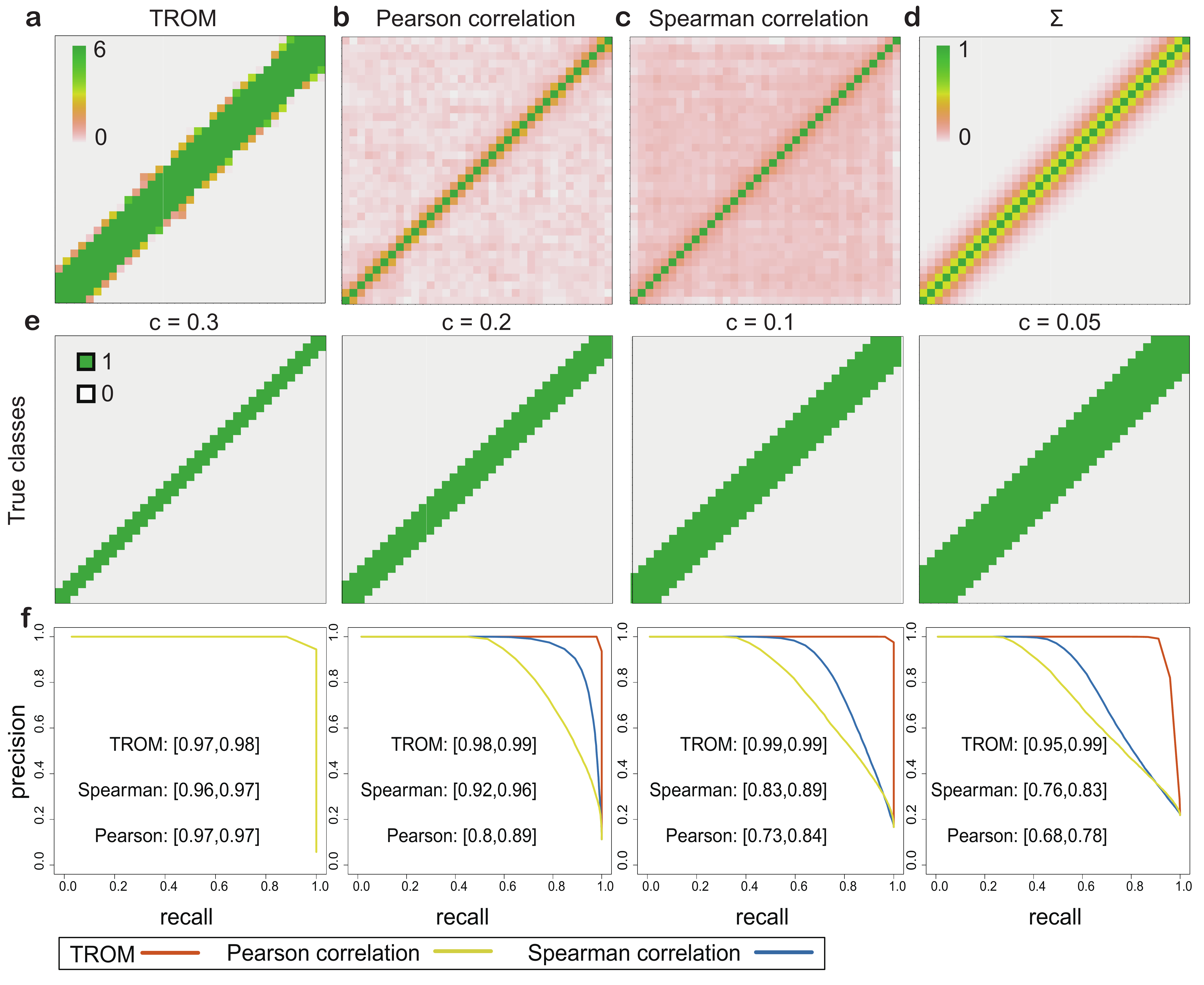}
\caption{ Comparison of TROM and Pearson/Spearman correlation on simulated \emph{C. elegans} (worm) data. \JL{\textbf{a-c}: The correspondence maps produced by TROM (\textbf{a}), Pearson correlation (\textbf{b}) and Spearman correlation (\textbf{c}) on a randomly selected gene expression matrix (among the $200$ matrices). \textbf{d}: The correlation matrix $\bm \Sigma$ that defines the dependence of associated genes between samples. \textbf{e}: The true sample relationships ($1$: high dependence in associated genes; $0$: otherwise)  defined as in Equation \ref{defineE} for varying $c$. \textbf{f}: The mean precision-recall curves on the $200$ gene expression matrices, given the true labels in \textbf{e}. The 95\% confidence intervals of each measure's area under the curve (AUC) are marked next to the curves.}}
\label{simu_worm}
\end{figure}

\JL{In addition, we} directly compare TROM with Pearson and Spearman correlation coefficients on the two real datasets of fly and worm used in the simulation. In our previous work \cite{Li2014}, we applied the preliminary idea of TROM to \JJL{compare} the developmental stages within each species and between the two species, and found interesting correspondence patterns: a \JJL{block} diagonal pattern for within-species comparison and two parallel patterns \JJL{between fly and worm developmental stages}. When using Pearson and Spearman correlations on the same data to compare these stages, however, \JL{we find that neither correlation measure leads to clear correspondence patterns} in the between-species comparison (Appendix Figure \ref{outline}). In the within-species comparison, Spearman correlation finds a vague diagonal pattern, while Pearson correlation leads to \JJL{an} unreasonable \JJL{checkerboard} pattern.
\noteW{
We also calculate Pearson and Spearman correlation matrices based on the union of all \JL{the} stage-associated genes \JL{found by TROM}. \JJL{However, correlation methods} still cannot provide \JL{clear} correspondence maps  \JL{like} \JJL{TROM} \JJL{does} (Appendix Figure \ref{cor_associated}).
}

\subsection {\noteW{Robustness of TROM to data normalization}}

Since quantile normalization has \JL{been suggested as} an essential step in \JL{many} analysis pipelines for high-throughput data such as microarray and RNA-seq \JJL{data} \cite{bolstad2003comparison,quantile2014}, we conduct a simulation study to demonstrate the influence of quantile normalization on TROM scores. We simulate $200$ gene expression matrices and compute their TROM scores \JL{with or without quantile normalization as a preceding step}. Then we \JL{test if} the distribution of TROM scores \JJL{changes with the use of} quantile normalization. 

We use the same procedure as \JL{what} described in Section \ref{sec:power_simu} to generate \JL{$200$} gene expression matrices based on \JL{the} modENCODE RNA-seq data of $35$ worm developmental stages. By applying the TROM method \JL{to these} gene expression \JL{matrices} before or after quantile normalization, we obtain two sets of TROM matrices $\bm T^{(0k)}$ and $\bm T^{(1k)}, k=1,2,\dots,200$. For each pair of samples, \JL{say samples $i$ and $j$}, we have two sets of TROM scores $T^{(0k)}_{ij}$ and $T^{(1k)}_{ij}$. We then use the Wilcoxon signed-rank test \JL{and separately the} paired Student's $t$ test to check whether the TROM scores change significantly before and after quantile normalization. We consider the change as significant if the Bonferroni-corrected $p$-value is smaller than $0.05$. The results are \JL{shown} in Figure \ref{simu_quant}.

\begin{figure}[!tpb]
\centering
\includegraphics[width=0.7\linewidth]{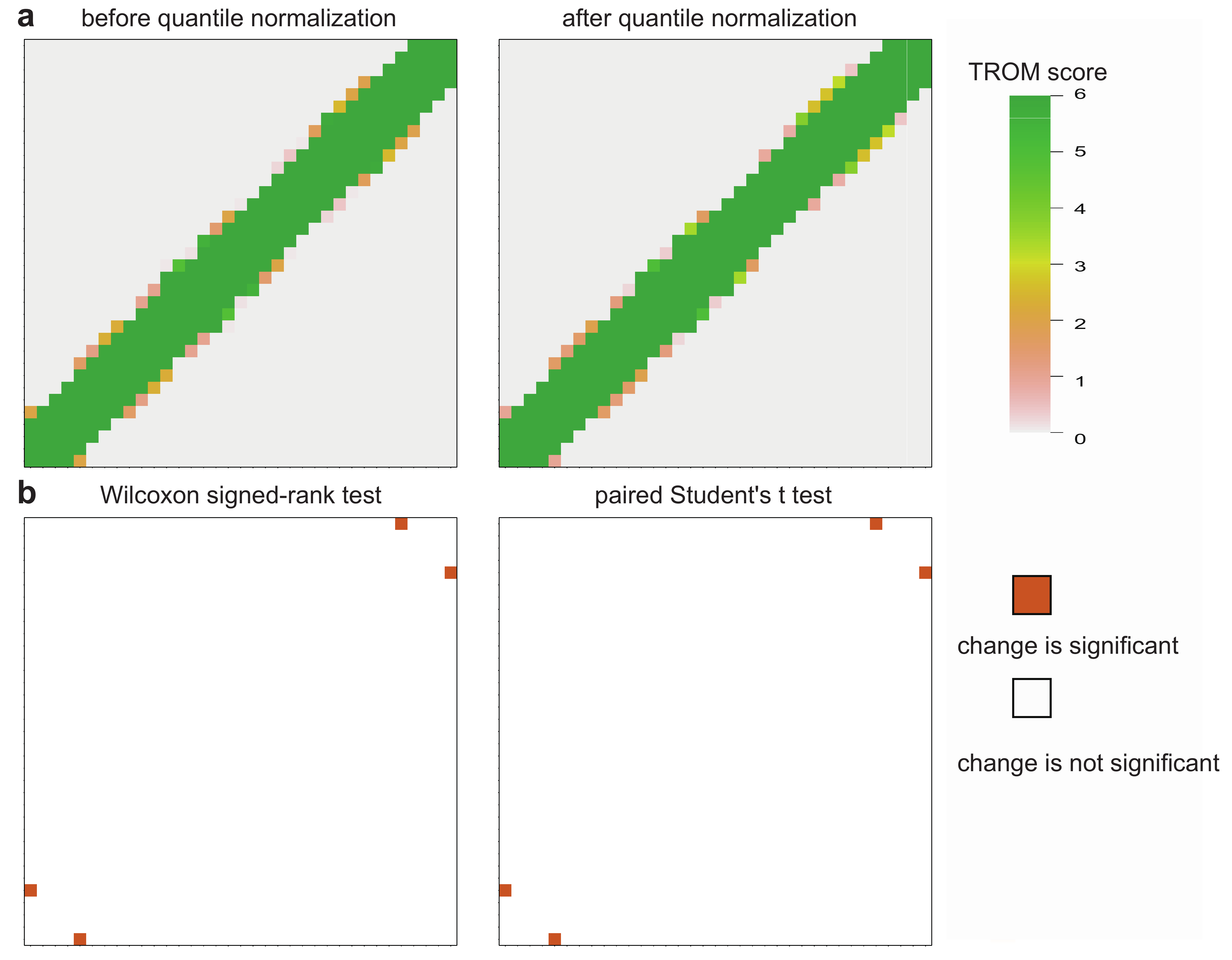}
\caption{Robustness of TROM to \JL{quantile} normalization on simulated \emph{C. elegans} (worm) data. \textbf{a}: The correspondence maps based on \JL{TROM scores of} a randomly selected \JL{gene expression} matrix (among the $200$ \JL{simulated matrices}), before (left) and after (right) quantile normalization.  \textbf{b}: The results of \JL{the Wilcoxon signed-rank test (left) and the paired Student's t test (right)}. \JL{Every blank cell means that the Bonferroni-corrected $p$-value is insignificant for the corresponding pair of stages, i.e., the TROM scores do not change significantly after quantile normalization.}}
\label{simu_quant}
\end{figure}

The results of both tests suggest that TROM is robust to \JL{unnormalized data}, and the correspondence patterns \JL{resulted from TROM scores do not change significantly after quantile normalization}. Even in the two rare cases where the $p$-values are significant \JL{(Figure \ref{simu_quant}b), the corresponding samples are consistently mapped} before and after normalization. We also \JL{try to replace the} gene expression data with \JL{their} normalized version in Section \ref{sec:power_simu}, and the confidence intervals of TROM's \JL{area under the curve (AUC)} remain the same. This result implies that the classification power of TROM is \JL{also} robust to data normalization.

\subsection {\noteJJLnew{Robustness of TROM to different platforms: comparison of \textit{D. melanogaster} developmental stages based on microarray and RNA-seq data}}
Although many studies have claimed that RNA-seq is the technique of choice that provides more accurate estimation of absolute gene expression levels compared with microarray \cite{zhao2014comparison,fu2009estimating}, several genome-wide analyses have also suggested that microarray can measure the expression of above-median expressed genes reasonably well, and on those genes the two platforms have good concordance \cite{wang2014concordance}. Since microarray has been widely used to study transcriptomes of multiple species under various conditions in the past decade, it is desirable to have a good comparative transcriptomic method that is robust to the platform difference of microarray and RNA-seq data. 

Here we demonstrate the robustness of TROM by applying it to \JJL{comparing} the microarray and RNA-seq data \JJL{of} the developmental stages of \textit{D. melanogaster}. If TROM is robust, it should identify strong correspondence between  similar developmental stages in the microarray and RNA-seq data. 
For a pair of developmental stages, one with microarray data and the other with RNA-seq data, TROM identifies a set of associated genes for each of them based on all the stages \JJL{with} microarray \JJL{and} RNA-seq data, respectively. Then TROM performs the overlap test and produces a correspondence map.
The results show that TROM can find almost perfect correspondence of the same \textit{D. melanogaster} embryonic stages between microarray or RNA-seq (Figure~\ref{TROM_map}c).
\JJL{There are five other \textit{Drosophila} species that} have similar developmental patterns as \textit{D. melanogaster}, as we have already shown in the within-species and between-species comparison in Section \ref{sec:TROM scores}. We also compare their microarray data \JJL{of} embryonic stages with the RNA-seq data of \textit{D. melanogaster} as a further check. In the result (Appendix Figure \ref{micro_vs_rna}), we observe strong \JJL{block} diagonal patterns. Although RNA-seq data \JL{contain} larvae, prepupae, and adult stages that do not have \JJL{corresponding} microarray data, the off-diagonal patterns, which we observe (1) between late embryos in microarray and prepupae in RNA-seq and (2) between early embryos in microarray and female adults in RNA-seq, are consistent with our previous within-species correspondence map based on RNA-seq data only \cite{Li2014} and previous studies \cite{Arbeitman2002}. These results show that TROM can find almost the same correspondence of \textit{Drosophila} developmental stages \JJL{regardless} of the platform being microarray or RNA-seq.

\subsection{Gene Ontology (GO) enrichment analysis}
To understand the biological functions behind the correspondence we \JJL{have observed} between developmental stages, we perform enrichment analysis \cite{ashburner2000gene} of biological process (BP) gene ontology (GO) terms in stage-associated genes, as a way to determine common biological functions and processes in corresponding stages.
First, we examine the GO term enrichment in the associated genes of every \emph{D. melanogaster} embryomnic stage, using RNA-seq data (with $z$-score threshold 1.5) and microarray data (with $z$-score threshold 0.5), respectively. The enrichment scores are defined as $-\log_{10}(\text{Bonferroni corrected}\ p\text{-value})$ where $p$-values are calculated based on the hypergeometric test, and the results are illustrated in Appendix Figures \ref{fly_GO_rna} \JL{and}  \ref{fly_GO_micro}. \noteJJLnew{For every fly embryonic stage, the} top 20 enriched GO terms \noteJJLnew{in the associated genes} identified \noteJJLnew{by} RNA-seq data \noteJJLnew{contain biological functions highly relevant to these stages, and many of these terms have been discovered as enriched in \noteJJLnew{relevant embryonic samples} by previous studies} \cite{Li2014,puniyani2010spex2}. A proportion of these top enrichment GO terms with \noteJJLnew{support in the literature} are listed in Table \ref{flyGO}. \noteJJLnew{The enriched GO terms identified from both RNA-seq and microarray data} support the \noteJJLnew{correspondence} patterns observed in TROM correspondence maps: common enriched GO terms are often shared by adjacent stages \JJL{whose pairwise TROM scores are high}. The top enriched GO terms \noteJJLnew{found by both microarray and RNA-seq are informative for further functional studies on the associated genes of every stage, so as to better understand embryonic development of \emph{D. melanogaster}}.
\begin{table}[t]
\caption{Selected enriched GO terms in each stage of \emph{D. melanogaster}.}{%
\resizebox{\columnwidth}{!}{%
\begin{tabular}{ll}
\hline
\noteJJL{Stage Name} & Top enriched GO terms  \\
\hline
Embryo 0-2h & oogenesis, DNA replication, germ cell development, neurogenesis\\
Embryo 2-4h & neurogenesis, mRNA splicing via spliceosome, zygotic determination of anterior/posterior axis\\
Embryo 4-6h & mRNA splicing via spliceosome, specification of segmental identity, cell fate specification\\
Embryo 6-8h & cell fate specification, sensory organ development, open tracheal system development\\
Embryo 8-10h & myoblast fusion, multicellular organism reproduction, puparial adhesion\\
Embryo 10-12h & myoblast fusion, translation, mitotic spindle elongation, septate junction assembly\\
Embryo 12-14h & axon guidance, septate junction assembly, branch fusion open tracheal system\\
Embryo 14-16h & circadian rhythm, response to light stimulus, crystal cell differentiation\\
Embryo 16-18h & chitin-based cuticle development, body morphogenesis, chitin metabolic process\\
Embryo 18-20h & body morphogenesis, chitin metabolic process, proteolysis\\
\hline
\end{tabular}
}
}
\label{flyGO}
\end{table}

\noteJJLnew{We also examine the GO term enrichment in the associated genes (identified with $z$-score threshold 1.5) of every developmental stage of mouse liver}. The resulting enrichment scores are illustrated in Appendix Figure \ref{mouse_GO}. The \noteJJLnew{top 10 enriched} GO terms in our selected stage-associated genes \noteJJLnew{of every stage} \noteJJLnew{confirm} previous findings on liver development and regeneration. In E11.5-12.5, \noteJJLnew{two of the early stages}, top enriched GO terms are \noteJJLnew{mostly} cell cycle related terms like \noteJJLnew{``translation," ``mRNA processing," ``cell cycle," and ``cell division"} \cite{li2009multi}. Previous research has shown that mouse liver takes over the function of hematopoiesis at E10.5-12.5 \cite{li2009multi,hata2007liver}\noteJJLnew{,} and we found that the GO terms \noteJJLnew{including ``heme biosynthetic process" and ``porphyrin-containing compound biosynthetic process"} are top enriched in subsequent stages. For \noteJJLnew{stages} E17.5-Day7, the GO terms \noteJJLnew{``innate immune response" and ``immune system process"} are top enriched, in accordance with the theory that liver is an organ with innate immune features \cite{dong2007roles}. Finally, as the function of mouse liver \noteJJLnew{switches} from hematopoiesis to metabolism and this capacity dominates in the adult liver \cite{hata2007liver,li2009multi}, we \noteJJLnew{observe} that GO terms related to various metabolic processes \noteJJLnew{become} enriched in \noteJJLnew{stages} E17.5-NL(normal adult liver tissue). These findings again illustrate the \noteJJLnew{capacity} of the associated genes in capturing \noteJJLnew{transcriptomic characteristics of biological samples}.

\subsection{Construction of developmental trees using stage-associated genes}
We further demonstrate that the selected stage-associated genes contain abundant information to group and distinct developmental stages. Tree construction has been a popular approach for studying the relationships of different developmental stages in organism development \cite{Arbeitman2002} as well as cell lineages in cell differentiation \cite{virmani2002hierarchical}. Here we attempt to construct developmental trees of diverse species (see Figure \ref{tree} and Appendix Figure \ref{tree_supp}) based on the identified associated genes of each developmental stage, reasoning that the associated genes capture stage characteristics and thus can lead to reasonable developmental trees. 
\noteW{
In tree construction, both Simpson and Jacard similarity coefficients can be used to measure the distance between the associated genes of different samples. However, Simpson coefficient will produce a result of $1$ when \JJL{the associated genes of one sample} is a subset of the \JJL{associated genes of the other sample, and it} thus fails to distinguish two samples in this case. In contrast, Jacard coefficient is able to separate two biological samples in this case\JJL{, because it} considers two samples as identical if and only if they have exactly the same associated genes.
}
As a consequence, we carry out the tree construction by hierarchical clustering, using average linkage and Jaccard coefficient, where the distance between two stages $i$ and $j$ is calculated as
\begin{equation}
J_{ij} = \frac{|X_i\cap X_j|}{|X_i| + |X_j| - |X_i\cap X_j|}\,,
\end{equation}
where $|X_i|$ \noteJJLnew{and} $|X_j|$ are the \noteJJLnew{sizes} of two sets of stage-associated genes and $|X_i\cap X_j|$ is the number of \noteJJLnew{genes in their} intersection.

\begin{figure}[!tpb]
\centering
\includegraphics[width=0.8\linewidth]{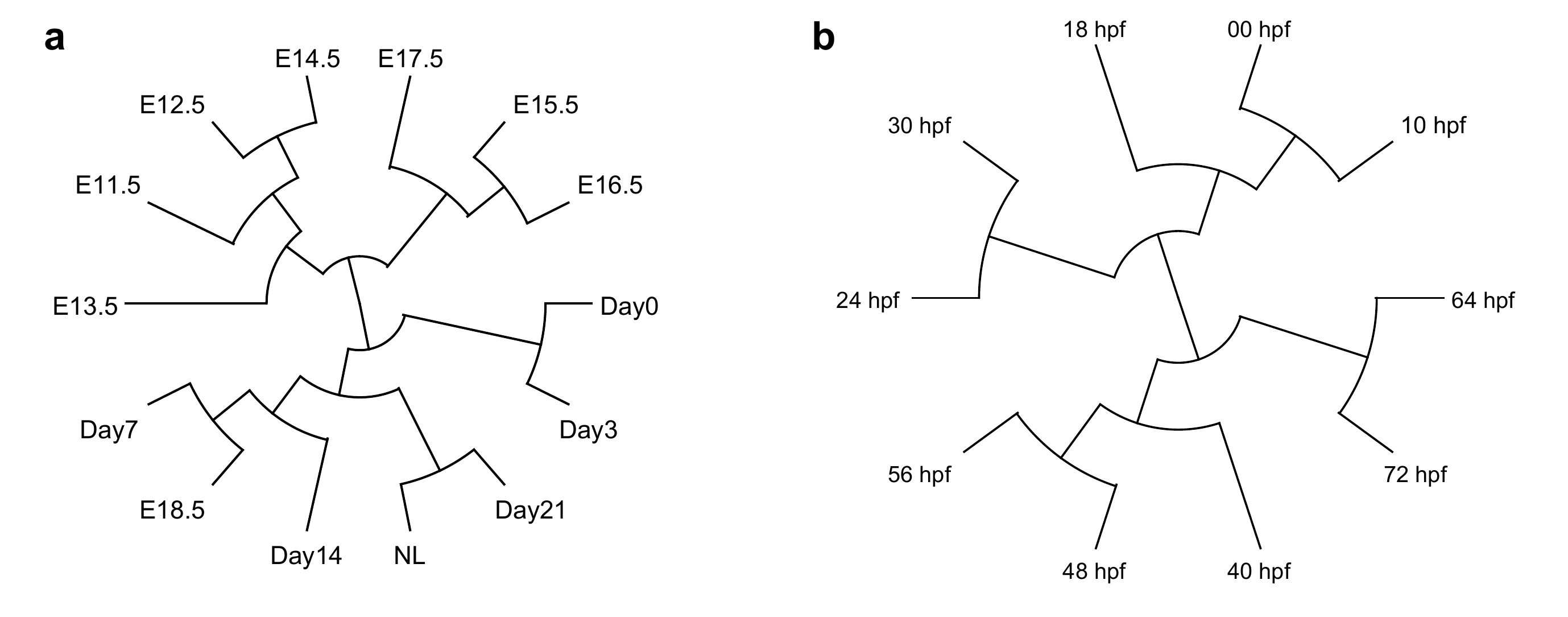}
\caption{\JJL{Developmental trees constructed based on} stage-associated genes (identified with $z$-score thresholds $1.4$ and $1.1$ for mouse liver and sea urchin respectively). \textbf{a}: \JJL{Developmental tree of} mouse liver. \textbf{b}: \JJL{Developmental tree of} sea urchin.}
\label{tree}
\end{figure}

The \noteJJLnew{developmental} tree (see Figure \ref{tree}a) constructed for mouse liver development shows an interesting pattern: the first major \noteJJLnew{branch} of the tree \noteJJLnew{successfully} divides the $14$ stages into embryonic stages and postnatal stages with one exception that the last embryonic stage E18.5 is clustered \noteJJLnew{with} the postnatal stages. \noteJJLnew{Moreover, neighboring stages are clustered with each other in small branches.} 
These observations are in accordance with the correspondence pattern illustrated by TROM scores (see Figure \ref{TROM_map}a): mappings exist between neighboring stages but not between E11.5-E17.5 and E18.5-NL. \noteJJLnew{Previous} hierarchical clustering results on genes whose expression levels are changed by more than 1.5-fold to average \cite{li2009multi} \noteJJLnew{supported our constructed tree and the} similarity between E18.5 and postnatal stages. The GO enrichment analysis \noteJJLnew{provides functional explanation on the observed clustering of E18.5 and Day 7, which both have enriched GO term including ``innate immune response," ``immune system process," and ``multicellular organismal development".}

The \noteJJLnew{developmental} tree (see Figure \ref{tree}b) constructed for sea urchin embryonic development also matches existent understanding of temporal interrelations of developmental \noteJJLnew{stages}. First, the major \noteJJLnew{branch} of the differentiation tree divides the stages into two sub-groups: one is 00, 10, 18, 24 \noteJJLnew{and} 30 hpf and the other is 40, 48, 56, 64 \noteJJLnew{and} 72 hpf. Previous studies show that oral/aboral (O/A) axis specification, endomesoderm development and autonomous specification are the major developmental processes before 40 hpf\noteJJLnew{,} while set-aside cells and rudiment formation and embryonic morphogenesis take over the major processes after 40 hpf \cite{davidson1998specific}. This functional explanation supports our constructed tree. Second, neighboring stages are grouped into small branches, and the overall tree is in accordance with \noteJJLnew{sea urchin's embryonic development periods} as cleavage, blastula, gastrula and prism-pluteus \cite{davidson1998specific}. 

We also observe reasonable and meaningful developmental trees constructed for the six \textit{Drosophila} species and \textit{C. elegans} (Appendix Fig. \ref{tree_supp}). We note that the tree construction is robust to the $z$-score threshold choices.

\section{Discussion}
\noteW{
In this work, we demonstrate that our proposed measure TROM is more efficient in finding transcriptomic similarity and correspondence patterns \JJL{of biological samples within and between} species compared with \JJL{Pearson and Spearman correlations}. Both simulation and real data analysis verify the superior power of TROM in detecting biologically meaningful relationships between different samples. The comparison results suggest that \JJL{in the TROM method} the selection of associated genes is a \JJL{critical} step \JJL{before} the overlap test. The selection step ensures that the transcriptomic characteristics of each sample are \JJL{well} captured and represented. \JJL{Moreover, the strength of TROM also lies} in the overlap test \JJL{that} does not directly rely on absolute gene expression values and is thus relatively robust to noisy data. On the other hand, \JJL{Pearson and Spearman correlations} fail to \JJL{detect clear} correspondence patterns even \JJL{based on the} associated genes. 
}

\noteW{
We \JJL{observe} that it is possible to improve the \JJL{correspondence map found by} Spearman correlation by \JJL{thresholding its correlation values, i.e., setting all the values below the threshold to the minimum value of all pairwise comparisons}. We test this procedure on the RNA-seq datasets of \emph{D. melanogaster} and \emph{C. elegans} and the results are summarized in Appendix Figure \ref{sp_thre}. As expected, thresholding on the Spearman correlation can give rise to relatively clearer correspondence patterns. However, this procedure is very sensitive to the threshold and often miss biologically meaningful mappings: the similarity of early embryos and female adults in fly is only captured once and the similarity of embryo and adults in worm is totally \JJL{missing at all thresholds} \cite{Li2014}.}

\noteW{
We would also like to point out that although TROM is not a parameter-free method, the resulting similarity patterns are largely robust to the selection of \JJL{the $z$-score threshold}. In addition, \JJL{the TROM method} provides users with the flexibility to tune the \JJL{threshold} according to \JJL{the level of} relationships they \JJL{look} for between biological samples.
}

\noteW{
 The sample-associated genes identified \JJL{based on} the threshold carry important transcriptomic characteristics of the corresponding samples and are not simply the complement of housekeeping genes. The identification of sample-associated genes filters out not only housekeeping genes, but also those genes that exhibit little variation across samples. In addition, it is worth noting that the concept of associated genes is not equivalent to specific genes, since associated genes also contain genes that capture transcriptomic similarity among closely related samples, and these genes can be shared \JJL{by several but not all samples}.}

\noteW{
\JJL{To the best of our knowledge,} Le et al \cite{le2010cross} is the only previous attempt other than correlation-based methods to compare biological samples across species. This method compares expression experiments from different species through a newly defined distance metric between the ranking of orthologous genes in the two species. However, their method relies on a large training dataset of known similar samples to learn the parameters for distance functions, and is thus not practical for finding novel patterns \JJL{of biological samples from rarely studied species such as \textit{D. rerio}}. \JJL{Another} advantage of TROM compared with this method is that TROM can identify informative associated genes that enable various downstream analyses.
}

\section{Conclusion}
TROM, a testing-based method, is introduced for finding correspondence patterns among transcriptomes of the same or different species. We demonstrate the greater power of TROM compared to correlation measures in finding transcriptomic similarity in terms of highly expressed genes. We apply TROM to find correspondence maps of developmental stages within and between multiple species, and we show that the associated genes TROM identifies for developmental stages can be used to construct developmental trees in these species. We also show that TROM is robust to data normalization and platform difference of microarray and RNA-seq. In addition, we design a systematic approach for selecting a key threshold parameter in TROM. We implement the TROM method in an R package, which provides functions with flexibility for illustration and customization and can be easily integrated into existing comparative genomic pipelines.


\bibliographystyle{spmpsci}      
\bibliography{TROM}   

\begin{thebibliography}{10}
\providecommand{\url}[1]{{#1}}
\providecommand{\urlprefix}{URL }
\expandafter\ifx\csname urlstyle\endcsname\relax
  \providecommand{\doi}[1]{DOI~\discretionary{}{}{}#1}\else
  \providecommand{\doi}{DOI~\discretionary{}{}{}\begingroup
  \urlstyle{rm}\Url}\fi

\bibitem{Arbeitman2002}
Arbeitman, M.N., Furlong, E.E., Imam, F., Johnson, E., Null, B.H., Baker, B.S.,
  Krasnow, M.A., Scott, M.P., Davis, R.W., White, K.P.: Gene expression during
  the life cycle of \textit{Drosophila melanogaster}.
\newblock Science \textbf{297}(5590), 2270--2275 (2002)

\bibitem{ashburner2000gene}
Ashburner, M., Ball, C.A., Blake, J.A., Botstein, D., Butler, H., Cherry, J.M.,
  Davis, A.P., Dolinski, K., Dwight, S.S., Eppig, J.T., et~al.: Gene ontology:
  tool for the unification of biology.
\newblock Nature genetics \textbf{25}(1), 25--29 (2000)

\bibitem{bolstad2003comparison}
Bolstad, B.M., Irizarry, R.A., {\AA}strand, M., Speed, T.P.: A comparison of
  normalization methods for high density oligonucleotide array data based on
  variance and bias.
\newblock Bioinformatics \textbf{19}(2), 185--193 (2003)

\bibitem{cunningham2015ensembl}
Cunningham, F., Amode, M.R., Barrell, D., Beal, K., Billis, K., Brent, S.,
  Carvalho-Silva, D., Clapham, P., Coates, G., Fitzgerald, S., et~al.: Ensembl
  2015.
\newblock Nucleic acids research \textbf{43}(D1), D662--D669 (2015)

\bibitem{davidson1998specific}
Davidson, E.H., Cameron, R.A., Ransick, A.: Specification of cell fate in the
  sea urchin embryo: summary and some proposed mechanisms.
\newblock Development \textbf{125}(17), 3269--3290 (1998)

\bibitem{domazet2010phylogenetically}
Domazet-Lo{\v{s}}o, T., Tautz, D.: A phylogenetically based transcriptome age
  index mirrors ontogenetic divergence patterns.
\newblock Nature \textbf{468}(7325), 815--818 (2010)

\bibitem{dong2007roles}
Dong, Z., Wei, H., Sun, R., Tian, Z.: The roles of innate immune cells in liver
  injury and regeneration.
\newblock Cell Mol Immunol \textbf{4}(4), 241--252 (2007)

\bibitem{fu2009estimating}
Fu, X., Fu, N., Guo, S., Yan, Z., Xu, Y., Hu, H., Menzel, C., Chen, W., Li, Y.,
  Zeng, R., et~al.: Estimating accuracy of {RNA-S}eq and microarrays with
  proteomics.
\newblock BMC genomics \textbf{10}(1), 161 (2009)

\bibitem{gerstein2014comparative}
Gerstein, M.B., Rozowsky, J., Yan, K.K., Wang, D., Cheng, C., Brown, J.B.,
  Davis, C.A., Hillier, L., Sisu, C., Li, J.J., et~al.: Comparative analysis of
  the transcriptome across distant species.
\newblock Nature \textbf{512}(7515), 445--448 (2014)

\bibitem{hata2007liver}
Hata, S., Namae, M., Nishina, H.: Liver development and regeneration: from
  laboratory study to clinical therapy.
\newblock Development, growth \& differentiation \textbf{49}(2), 163--170
  (2007)

\bibitem{quantile2014}
Hicks, S.C., Irizarry, R.A.: When to use quantile normalization?
\newblock bioRxiv p. 012203 (2014)

\bibitem{Labbe2012}
Labb{\'e}, R.M., Irimia, M., Currie, K.W., Lin, A., Zhu, S.J., Brown, D.D.,
  Ross, E.J., Voisin, V., Bader, G.D., Blencowe, B.J., et~al.: A comparative
  transcriptomic analysis reveals conserved features of stem cell pluripotency
  in planarians and mammals.
\newblock Stem Cells \textbf{30}(8), 1734--1745 (2012)

\bibitem{le2010cross}
Le, H.S., Oltvai, Z.N., Bar-Joseph, Z.: Cross-species queries of large gene
  expression databases.
\newblock Bioinformatics \textbf{26}(19), 2416--2423 (2010)

\bibitem{Li2014}
Li, J.J., Huang, H., Bickel, P.J., Brenner, S.E.: Comparison of \textit{D.
  melanogaster} and \textit{C. elegans} developmental stages, tissues, and
  cells by modencode {RNA-s}eq data.
\newblock Genome research \textbf{24}(7), 1086--1101 (2014)

\bibitem{li2009multi}
Li, T., Huang, J., Jiang, Y., Zeng, Y., He, F., Zhang, M.Q., Han, Z., Zhang,
  X.: Multi-stage analysis of gene expression and transcription regulation in
  c57/b6 mouse liver development.
\newblock Genomics \textbf{93}(3), 235--242 (2009)

\bibitem{Necsulea2014}
Necsulea, A., Soumillon, M., Warnefors, M., Liechti, A., Daish, T., Zeller, U.,
  Baker, J.C., Gr{\"u}tzner, F., Kaessmann, H.: The evolution of {lncRNA}
  repertoires and expression patterns in tetrapods.
\newblock Nature \textbf{505}(7485), 635--640 (2014)

\bibitem{Pantalacci2015}
Pantalacci, S., S{\'e}mon, M.: Transcriptomics of developing embryos and
  organs: a raising tool for evo--devo.
\newblock Journal of Experimental Zoology Part B: Molecular and Developmental
  Evolution \textbf{324}(4), 363--371 (2015)

\bibitem{puniyani2010spex2}
Puniyani, K., Faloutsos, C., Xing, E.P.: Spex2: automated concise extraction of
  spatial gene expression patterns from fly embryo {ISH} images.
\newblock Bioinformatics \textbf{26}(12), i47--i56 (2010)

\bibitem{Shen2012}
Shen, Y., Yue, F., McCleary, D.F., Ye, Z., Edsall, L., Kuan, S., Wagner, U.,
  Dixon, J., Lee, L., Lobanenkov, V.V., et~al.: A map of the cis-regulatory
  sequences in the mouse genome.
\newblock Nature \textbf{488}(7409), 116--120 (2012)

\bibitem{Spencer2011}
Spencer, W.C., Zeller, G., Watson, J.D., Henz, S.R., Watkins, K.L., McWhirter,
  R.D., Petersen, S., Sreedharan, V.T., Widmer, C., Jo, J., et~al.: A spatial
  and temporal map of c. elegans gene expression.
\newblock Genome research \textbf{21}(2), 325--341 (2011)

\bibitem{nproc}
Tong, X., Feng, Y., Li, J.J.: Neyman-{P}earson ({NP}) classification algorithms
  and {NP} receiver operating characteristic ({NP-ROC}) curves.
\newblock arXiv preprint arXiv:1608.03109  (2016)

\bibitem{tu2014quantitative}
Tu, Q., Cameron, R.A., Davidson, E.H.: Quantitative developmental
  transcriptomes of the sea urchin \textit{Strongylocentrotus purpuratus}.
\newblock Developmental biology \textbf{385}(2), 160--167 (2014)

\bibitem{virmani2002hierarchical}
Virmani, A.K., Tsou, J.A., Siegmund, K.D., Shen, L.Y., Long, T.I., Laird, P.W.,
  Gazdar, A.F., Laird-Offringa, I.A.: Hierarchical clustering of lung cancer
  cell lines using {DNA} methylation markers.
\newblock Cancer Epidemiology Biomarkers \& Prevention \textbf{11}(3), 291--297
  (2002)

\bibitem{wang2014concordance}
Wang, C., Gong, B., Bushel, P.R., Thierry-Mieg, J., Thierry-Mieg, D., Xu, J.,
  Fang, H., Hong, H., Shen, J., Su, Z., et~al.: The concordance between
  {RNA-seq} and microarray data depends on chemical treatment and transcript
  abundance.
\newblock Nature biotechnology \textbf{32}(9), 926--932 (2014)

\bibitem{Wang2009}
Wang, Z., Gerstein, M., Snyder, M.: {RNA-Seq}: a revolutionary tool for
  transcriptomics.
\newblock Nature Reviews Genetics \textbf{10}(1), 57--63 (2009)

\bibitem{zhao2014comparison}
Zhao, S., Fung-Leung, W.P., Bittner, A., Ngo, K., Liu, X.: Comparison of
  {RNA-Seq} and microarray in transcriptome profiling of activated {T} cells.
\newblock PloS One \textbf{9}(1) (2014)

\end{thebibliography}

\newpage
\section*{Appendix}
\setcounter{table}{0}
\renewcommand{\thetable}{A\arabic{table}}
\renewcommand{\theHtable}{A\arabic{table}}
\setcounter{figure}{0}
\renewcommand{\thefigure}{A\arabic{figure}}
\renewcommand{\theHfigure}{A\arabic{figure}}

\begin{longtable}{cc}
\caption{Selected $z$-score threshold for different species}\\

\hline
\noteJJL{Species} & threshold of $z$-scores  \\
\hline
\emph{D. melanogaster} (RNA-seq) & $1.8$ \\
\emph{C. elegans} &  $2.0$\\
\emph{D. melanogaster} (microarray) & $0.9$\\ 
\emph{D. ananassae} & $0.9$ \\
\emph{D. simulans} & $0.9$ \\
\emph{D. persimilis} & $0.9$ \\
\emph{D. pseudoobscura} &  $0.8$\\
\emph{D. virilis} & $1.1$ \\
mouse liver & $1.4$\\
sea urchin & $1.1$\\
\emph{D. rerio} & $1.0$\\
\hline
\label{tab:z_thre}
\end{longtable}

\begin{longtable}{p{0.2\textwidth}p{0.8\textwidth}}
\caption{Description of sample labels of different species}\\
\hline
\noteJJL{Species} & Sample labels and corresponding explanation  \\
\hline
\emph{D. melanogaster} (RNA-seq) & Embryo0-2h, Embryo2-4h, Embryo4-6h, Embryo6-8h, Embryo8-10h, Embryo10-12h, Embryo12-14h, Embryo14-16h, Embryo16-18h, Embryo18-20h, Embryo20-22h, Embryo22-24h, L1 (L1 stage larvae), L2 (L2 stage larvae), L3+12h (L3 stage larvae, 12 hr post-molt), L3PS1-2 (L3 stage larvae, dark blue gut, puff stage 1-2), L3PS3-6 (L3 stage larvae, light blue gut, puff stage 3-6), L3PS7-9 (L3 stage larvae, clear gut puff stage 7-9), Prepupae (White prepupae), Prepupae+12h (Pupae, 12 hours after white prepupae), Prepupae+24h (Pupae, 24 hours after white prepupae), Prepupae+2d (Pupae, 2 days after white prepupae), Prepupae+3d (Pupae, 3 days after white prepupae), Prepupae+4d (Pupae, 4 days after white prepupae), Male+1d (Adult male, one day after eclosion), Male+5d (Adult male, 5 days after eclosion), Female+1d (Adult female, one day after eclosion), Female+5d (Adult female, 5 days after eclosion), Female+30d (Adult female, 30 days after eclosion) \\
\hline
\emph{C. elegans} & EE\_50-0(embryo 0 mins), EE\_50-30 (embryo 30 mins), EE\_50-60, EE\_50-90, EE\_50-120, EE\_50-150, EE\_50-180, EE\_50-210, EE\_50-240, EE\_50-300, EE\_50-330, EE\_50-360, EE\_50-390, EE\_50-420, EE\_50-450, EE\_50-480, EE\_50-510, EE\_50-540, EE\_50-570, EE\_50-600, EE\_50-630, EE\_50-660, EE\_50-690, EE\_50-720, L1 (larva L1), LIN35 (larva L1 lin35), L2 (larva L2), L3 (larva L3), L4 (larva L4), L4MALE (larva L4 male), YA (young adult), AdultSPE9 (adult spe9), DauerEntryDAF2, DauerDAF2, DauerExitDAF2 \\
\hline
\emph{Drosophila} (microarray) & E0-2h (Embryo 0-2h), E2-4h (Embryo 2-4h), E4-6h (Embryo 4-6h), E6-8h (Embryo 6-8h), E8-10h (Embryo 8-10h), E10-12h (Embryo 12-14h), E14-16h (Embryo 14-16h), E16-18h (Embryo 16-18h), E18-20h (Embryo 18-20h), E20-22h (Embryo 20-22h), E22-24h (Embryo 22-24h), E22-24h (Embryo 24-26h)\\
\hline
mouse liver & E11.5 (embryonic day 11.5), E12.5 (embryonic day 12.5), E13.5 (embryonic day 13.5), E14.5 (embryonic day 14.5), E15.5 (embryonic day 15.5), E16.5 (embryonic day 16.5), E17.5 (embryonic day 17.5), E18.5 (embryonic day 18.5), Day0 (the day of birth), Day3, Day7, Day14, Day21, and NL (normal adult liver)\\
\hline
sea urchin & 00 hpf (0 hours post-fertilization), 10 hpf, 18 hpf, 24 hpf, 30 hpf, 40 hpf, 48 hpf, 56 hpf, 64 hpf, 72 hpf\\
\hline
\emph{D. rerio} & 0min (egg 0min), 15min (zygote 15min), 45min(cleavage 45min), 1h15min(cleavage 1h15min), 1h45min(cleavage 1h45min), 2h15min(blastula 2h15min), 2h45min(blastula 2h45min), 3h20min(blastula 3h20min), 4h(blastula 4h), 4h40min(blastula 4h40min), 5h20min(gastrula 5h20min), 6h(gastrula 6h), 7h(gastrula 7h), 8h(gastrula 8h), 9h(gastrula 9h), 10h(gastrula 10h), 10h20min(segmentation 10h20min), 11h(segmentation 11h), 11h40min(segmentation 11h40min), 12h(segmentation 12h), 13h(segmentation 13h), 14h(segmentation 14h), 15h(segmentation 15h), 16h(segmentation 16h), 17h(segmentation 17h), 18h(segmentation 18h), 19h(segmentation 19h), 20h(segmentation 20h), 21h(segmentation 21h), 22h(segmentation 22h), 23h(segmentation 23h), 1d1h(pharyngula 1d1h), 1d3h(pharyngula 1d3h), 1d6h(pharyngula 1d6h), 1d10h(pharyngula 1d10h), 1d14h(pharyngula 1d14h), 1d18h(pharyngula 1d18h), 2d(hatching 2d), 2d12h(hatching 2d12h), 3d(hatching 3d), 4d(larva 4d), 6d(larva 6d), 8d(larva 8d), 10d(larva 10d), 14d(larva 14d), 18d(larva 18d), 24d(larva 24d), 30d(larva 30d), 40d(larva 40d), 45d(juvenile 45d), 55d(juvenile 55d), 65d(juvenile 65d), 80d(juvenile 80d), 90d(adult 90d female), 3m15d(adult 3m15d female), 4m(adult 4m female), 7m(adult 7m female), 9m(adult 9m female), 1y2m(adult 1y2m female), 1y6m(adult 1y6m female), 1y9m(adult 1y9m), 55d(adult 55d male), 65d(adult 65d male), 80d(adult 80d male), 90d(adult 90d male), 3m15d(adult 3m15d male), 4m(adult 4m male), 7m(adult 7m male), 9m(adult 9m male), 1y2m(adult 1y2m male), 1y6m(adult 1y6m male), 1y9m(adult 1y9m male)\\
\hline
\label{stage_label}
\end{longtable}

\begin{figure}[!b]
\centerline{\includegraphics[width=\linewidth, angle=0, scale=1.3]{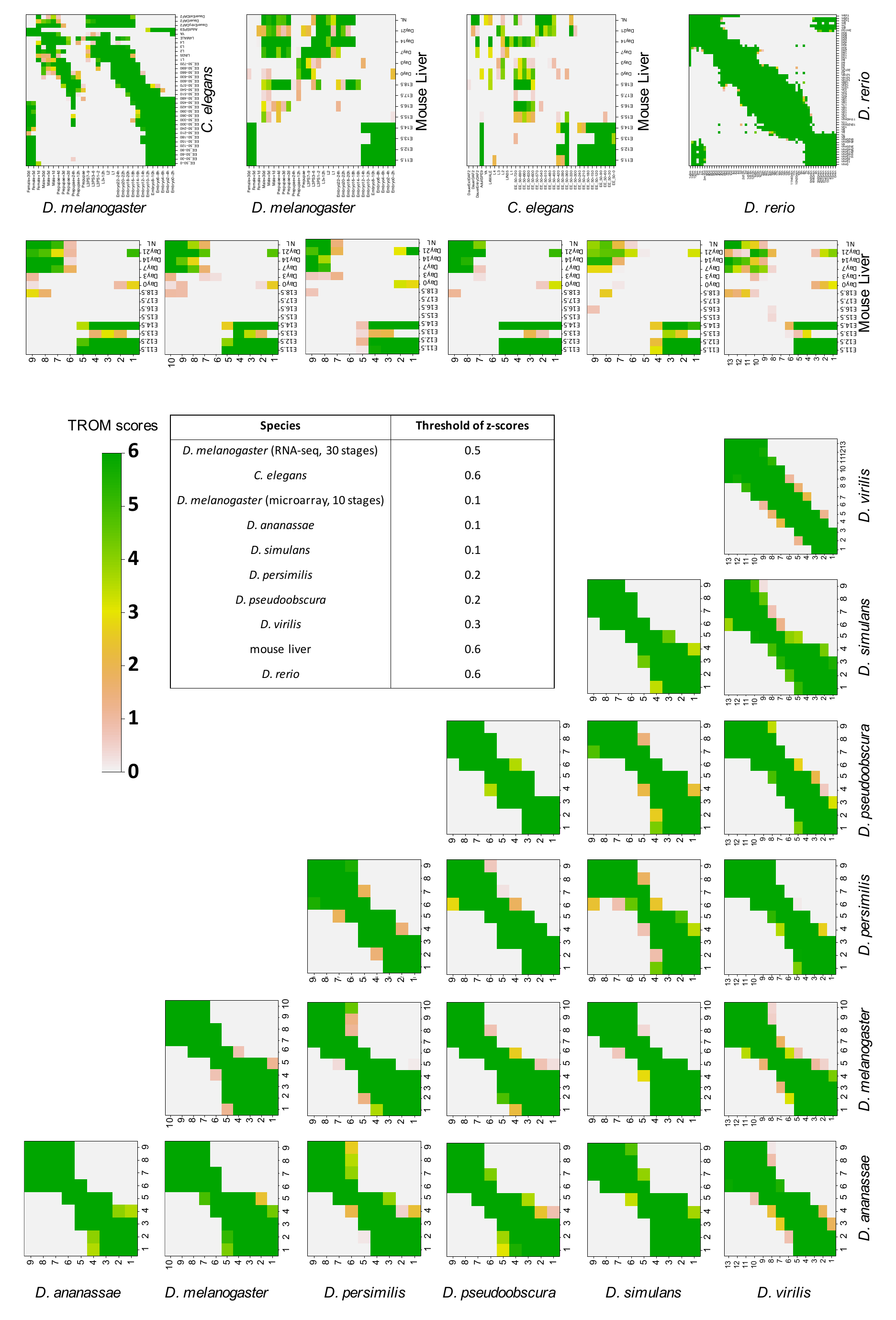}}
\caption{Correspondence maps of within-species and between-species TROM scores (calculated based on the $z$-score thresholds listed in the table). TROM scores are saturated at $6$. The names of the species are marked as row or column labels of the corresponding heatmaps. For the \emph{Drosophila} species the stages labels 1-13 refer to Embryo 0-2h, 2-4h, 4-6h, 6-8h, 8-10h, 10-12h, 12-14h, 14-16h, 16-18h, 18-20h, 20-22h, 22-24h and 24-26h respectively.}
\label{TROM_scores}
\end{figure}

\begin{figure}[!pb]
\centerline{\includegraphics[width=\linewidth]{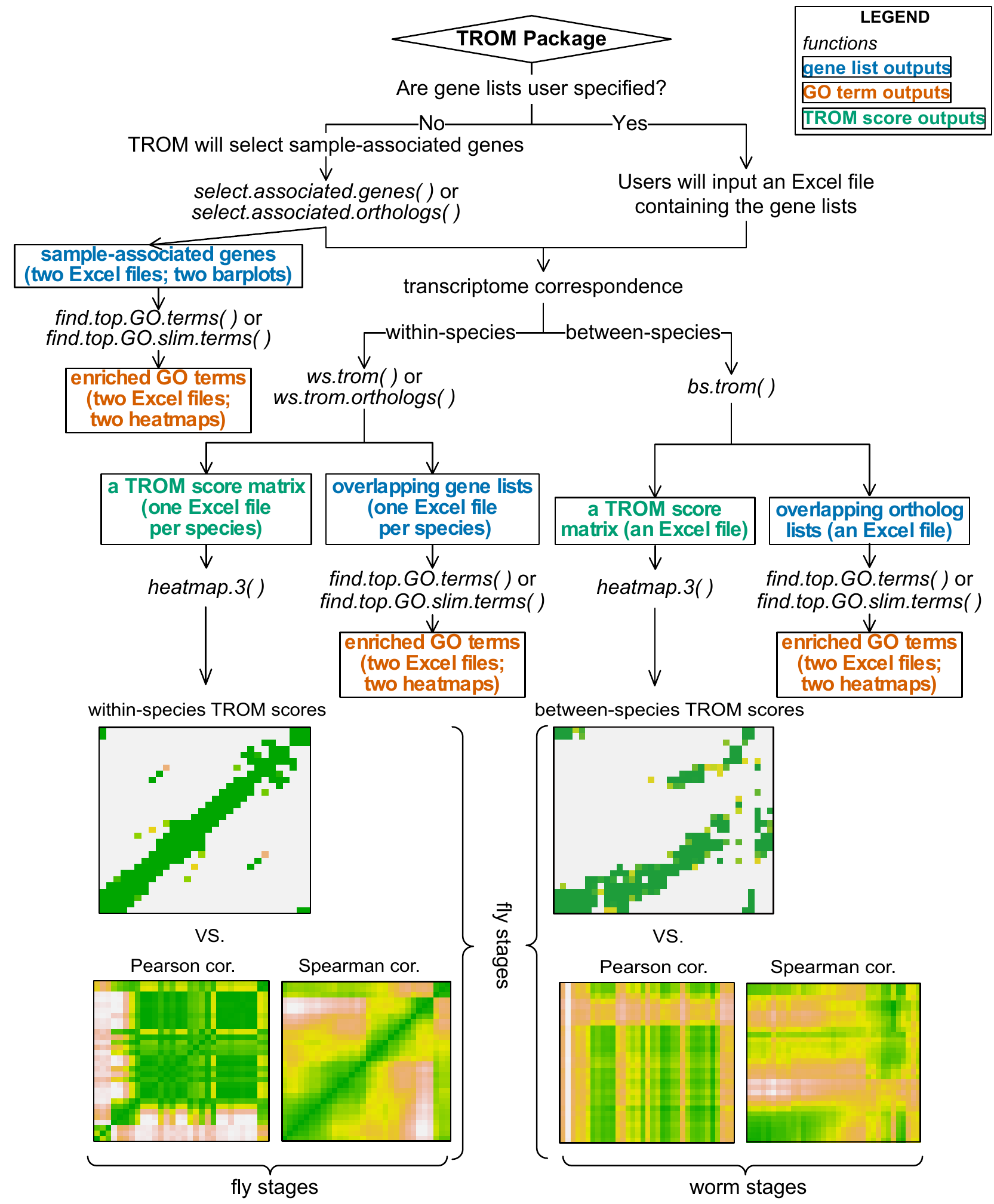}}
\cprotect\caption{Outline of the package \verb|TROM|. The three left (and right) heatmaps illustrate the within-species (and between-species) comparison results by using TROM (with $z$-score threshold 1.5 for both fly and worm), Pearson correlation and Spearman correlation. Greater similarities are shown in darker colors. The results show that compared to the popular Pearson and Spearman correlations, TROM can find clearer correspondence patterns. \verb|TROM| takes gene expression matrices and orthologous genes of the species of interest as input. The functions \verb|select.associated.genes| and \verb|select.associated.|\ \verb|orthologs| select the associated genes of different biological samples among all the genes or only among the genes with orthologs in the other species to be compared with. They also provide graphical summaries of the numbers of selected associated genes and orthologs. The functions \verb|ws.trom| and \verb|ws.trom.orthologs| perform the within-species transcriptome comparison, find the overlapping associated genes between every two samples and calculate within-species TROM scores. The function \verb|bs.trom| performs the between-species transcriptome comparison, find the overlapping associated orthologs between every two samples from different species and calculate the between-species TROM scores. The function \verb|heatmap.3| visualizes the TROM scores in a heatmap, with various add-on options for customization. The functions \verb|find.top.GO.terms| and \verb|find.top.GO.slim.terms| perform gene set enrichment analysis and find top enriched Gene Ontology (GO) terms and GO slim terms in the associated genes. Instead of using the selected associated genes, users may input customized gene lists representing characteristics of different biological samples into the above functions. Please see the package manual and vignette of \verb|TROM| for details.}
\label{outline}
\end{figure}

\begin{figure}[!pb]
\centerline{\includegraphics[width=\linewidth]{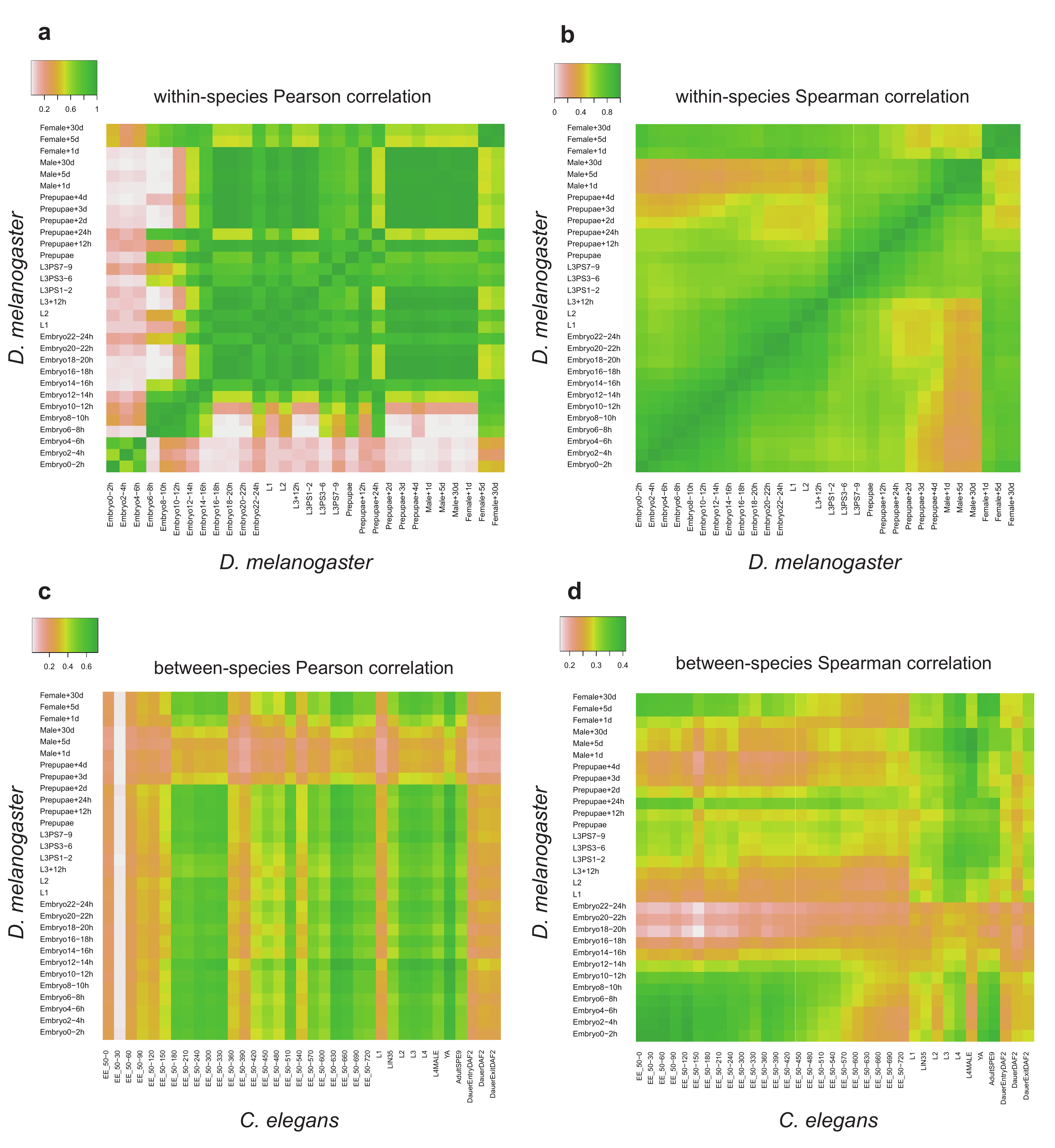}}
\caption{Correlation measures calculated based on the union of associated genes. \textbf{a-b:} Pearson correlation (\textbf{a}) and Spearman correlation (\textbf{b}) for every pair of \emph{D. melanogaster} stages calculated based on the union of associated genes of all stages. 
\textbf{c-d:} Pearson correlation (\textbf{c}) and Spearman correlation (\textbf{d}) for every pair of \emph{D. melanogaster} and \emph{C. elegans} stages calculated based on the union of associated ortholog pairs of all stages.  These heatmaps show that correlation measures calculated based on associated genes only still cannot lead to clear correspondence patterns.
}
\label{cor_associated}
\end{figure}

\begin{figure}[!tpb]
\centering
\includegraphics[width=\linewidth]{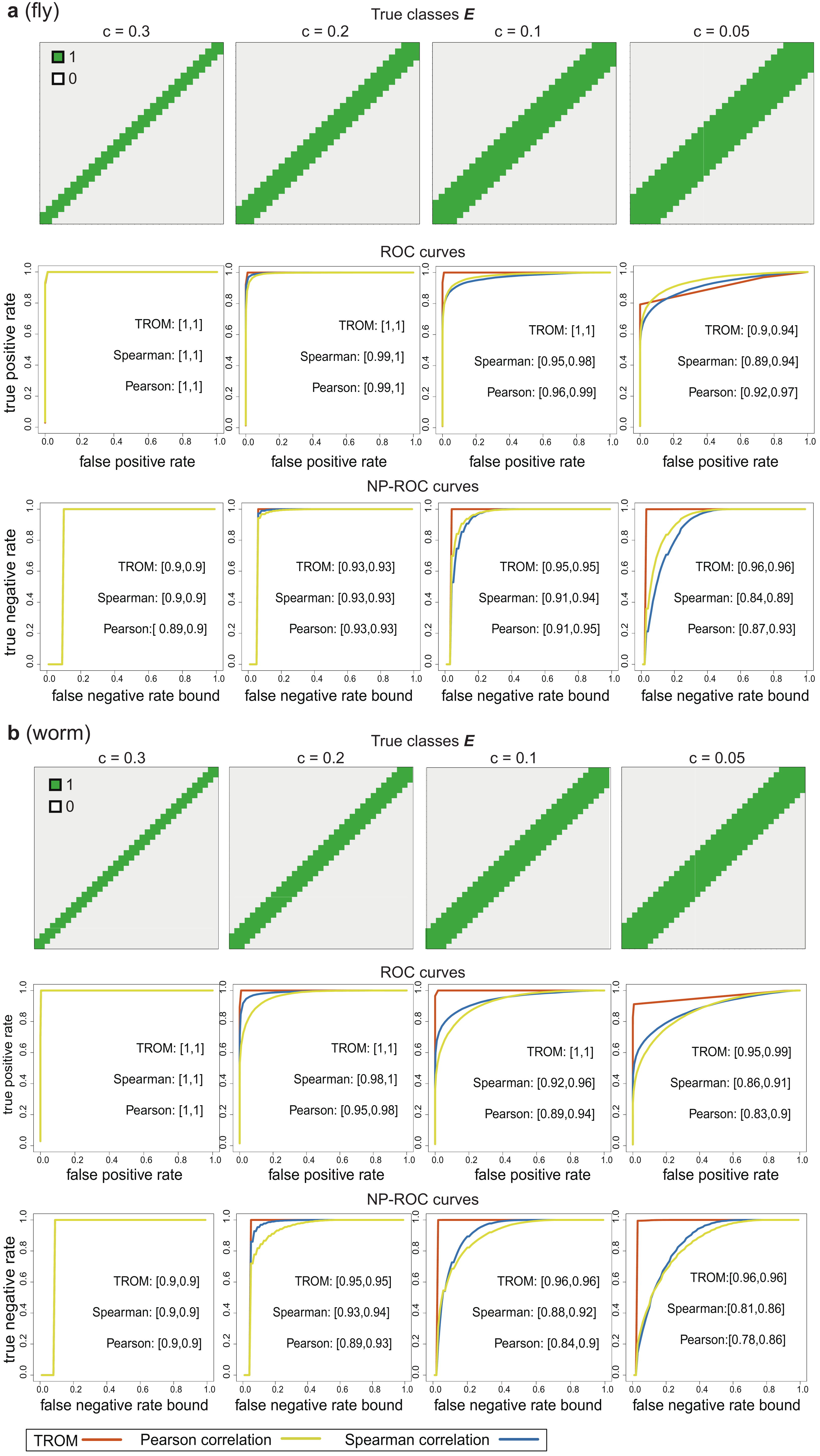}
\caption{Comparison of TROM and Pearson/Spearman correlation on simulated data, with \textbf{a} for fly and \textbf{b} for worm. \JL{In both panels, the first row gives the true sample relationships ($1$: high dependence in associated genes; $0$: otherwise)  defined as in Equation \ref{defineE} for varying $c$. The second row gives the mean receiver operating characteristic (ROC) curves on the $200$ simulated gene expression matrices, given the true labels in the first row. The third row gives the mean Neyman-Pearson receiver operating characteristic (NP-ROC) curves, accordingly. The 95\% confidence intervals of the area under the curve (AUC) are marked next to the curves.}}
\label{simu_roc}
\end{figure}

\begin{figure}[!pb]
\centerline{\includegraphics[width=\linewidth]{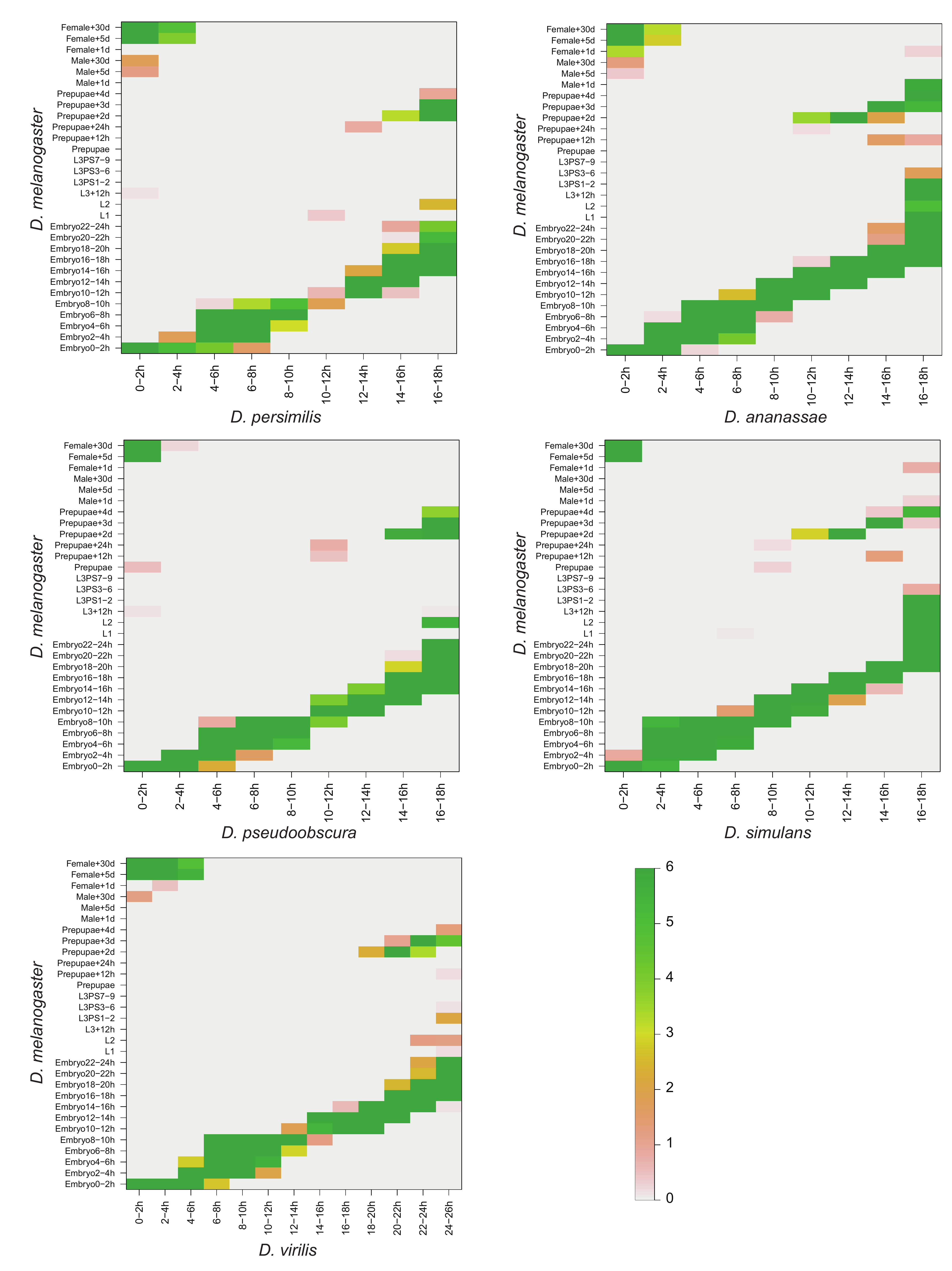}}
\caption{Correspondence maps of developmental stages. TROM scores are calculated using the RNA-seq data of \emph{D. melanogaster} and the microarray data of the other five \emph{Drosophila} species.}
\label{micro_vs_rna}
\end{figure}

\begin{figure}[!pb]
\centerline{\includegraphics[width=\linewidth]{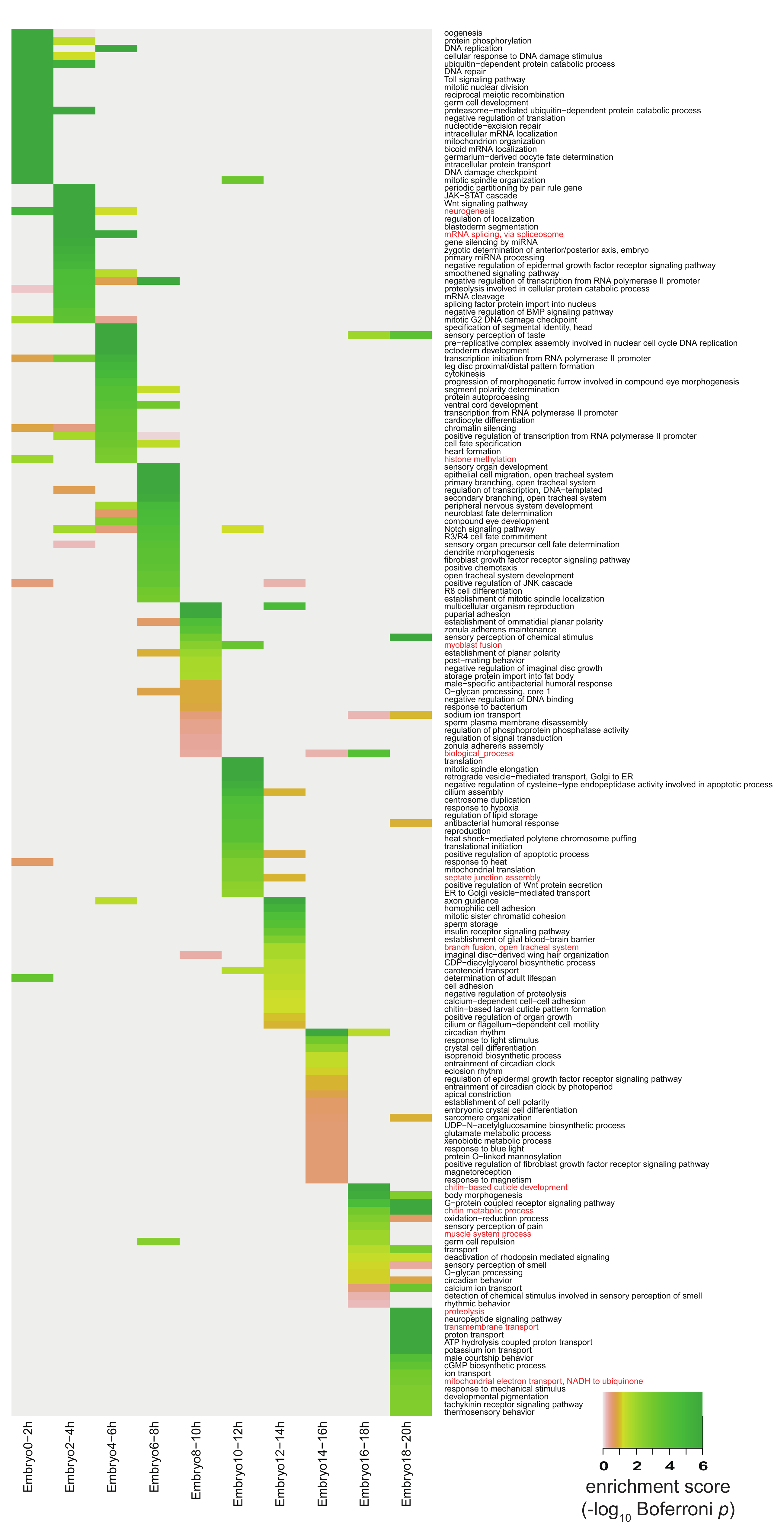}}
\caption{Top 20 enriched biological process GO terms of \textit{D.melanogaster}. The enrichment scores in the heatmap are calculated based on stage-associated genes identified from the RNA-seq data (with $z$-score threshold 1.5) and saturated at 6. For each stage, the common enriched GO terms identified from both microarray (Figure \ref{fly_GO_micro}) and RNA-seq datasets are marked in red color. }
\label{fly_GO_rna}
\end{figure}

\begin{figure}[!tpb]
\centerline{\includegraphics[width=\linewidth]{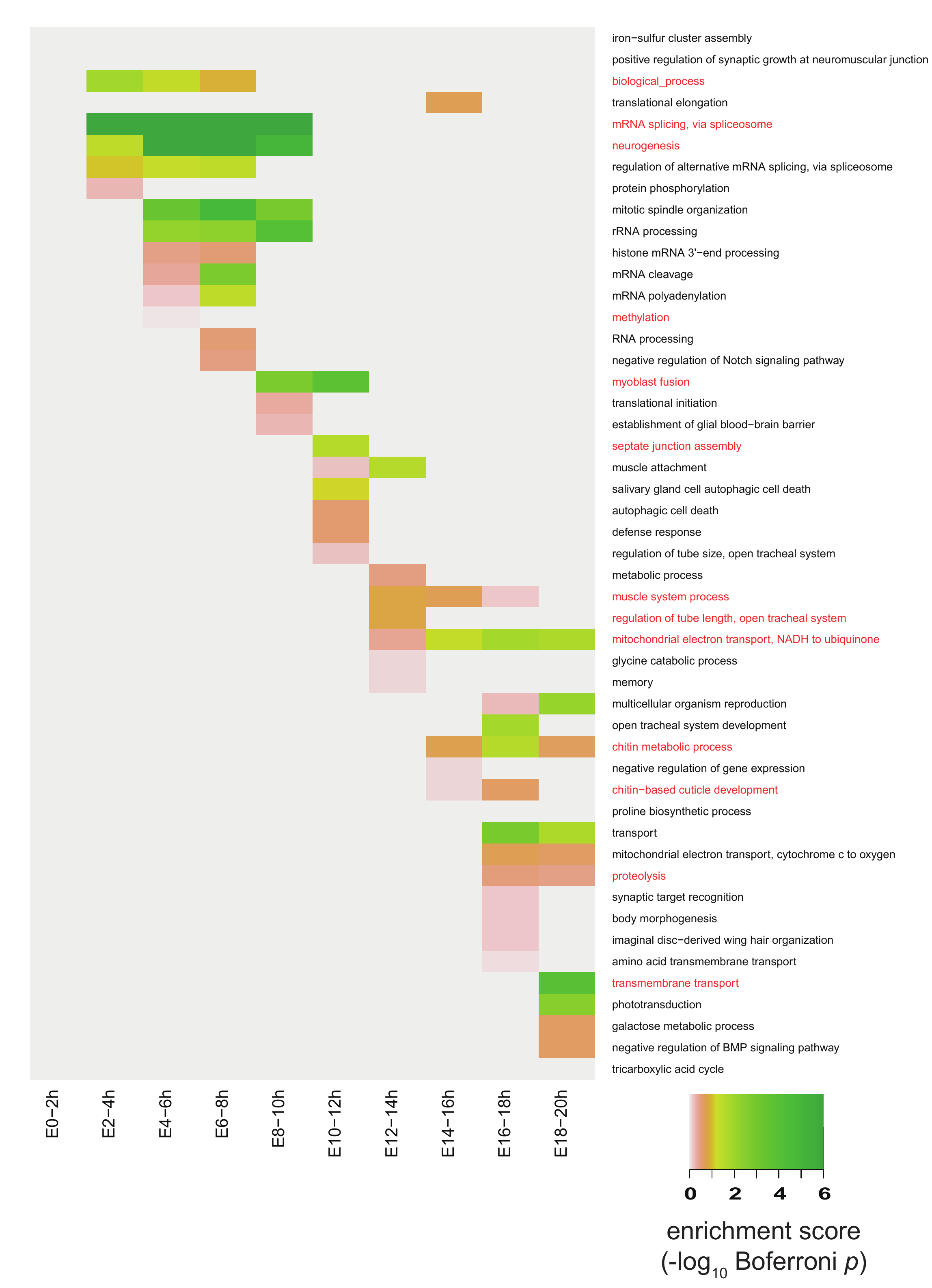}}
\caption{Top 20 enriched biological process GO terms of \textit{D.melanogaster}. The enrichment scores in the heatmap were calculated through stage-associated genes identified from the microarray data (with $z$-score threshold 0.5) and saturated at 6. For each stage, the common enriched GO terms identified from both microarray and RNA-seq (Figure \ref{fly_GO_rna}) datasets are marked in red color.  }
\label{fly_GO_micro}
\end{figure}

\begin{figure}[!pb]
\centerline{\includegraphics[width=\linewidth, scale =1.5]{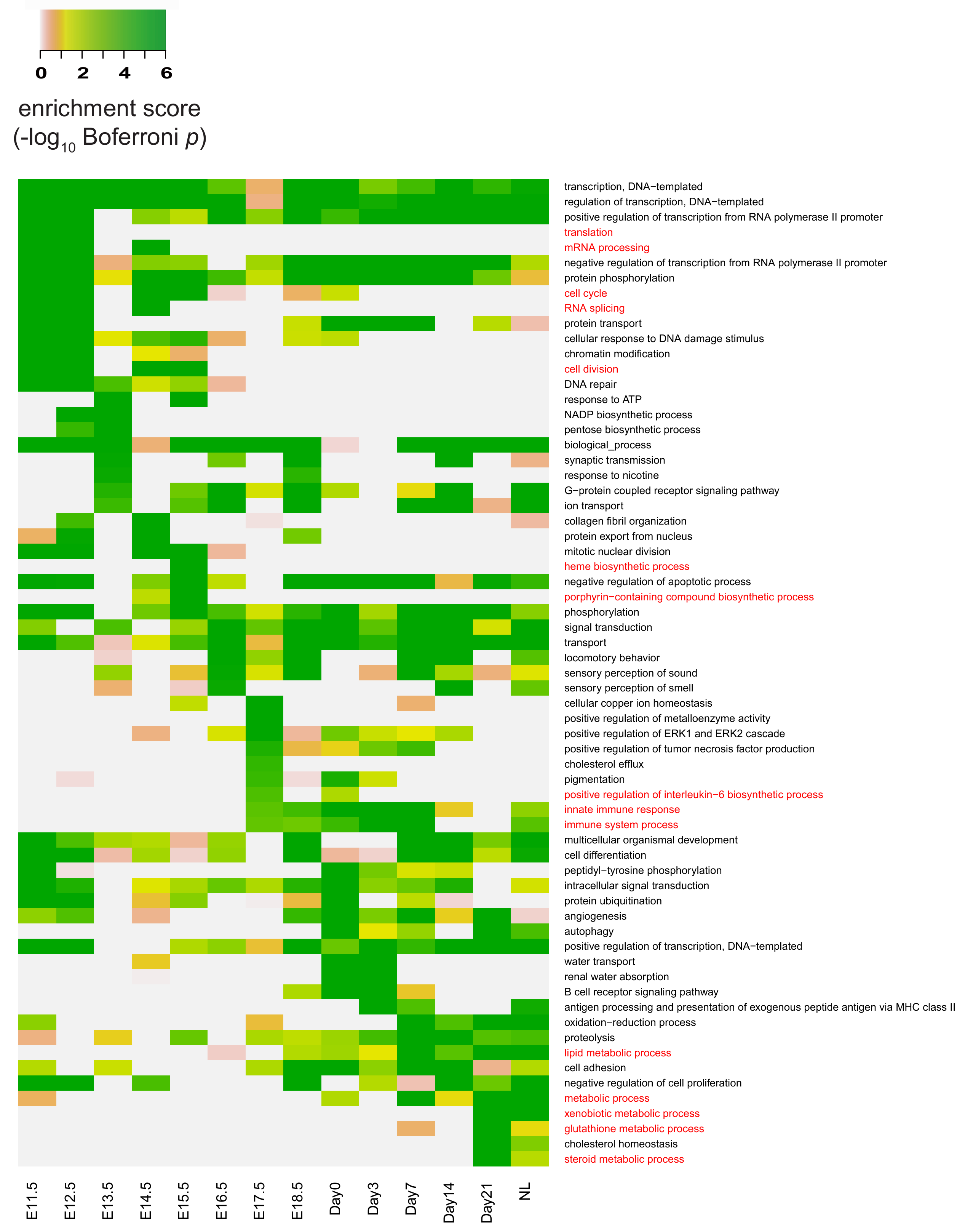}}
\caption{Top 10 enriched biological process GO terms of mouse liver. The enrichment scores in the heatmap 
 were calculated through stage-associated genes identified from the microarray data (with $z$-score threshold 1.5). For each stage, the highly 
 relevant GO terms that have been confirmed in previous studies are marked in red color. }
\label{mouse_GO}
\end{figure}

\begin{figure}[!pb]
\centerline{\includegraphics[width=\linewidth, scale=1.3]{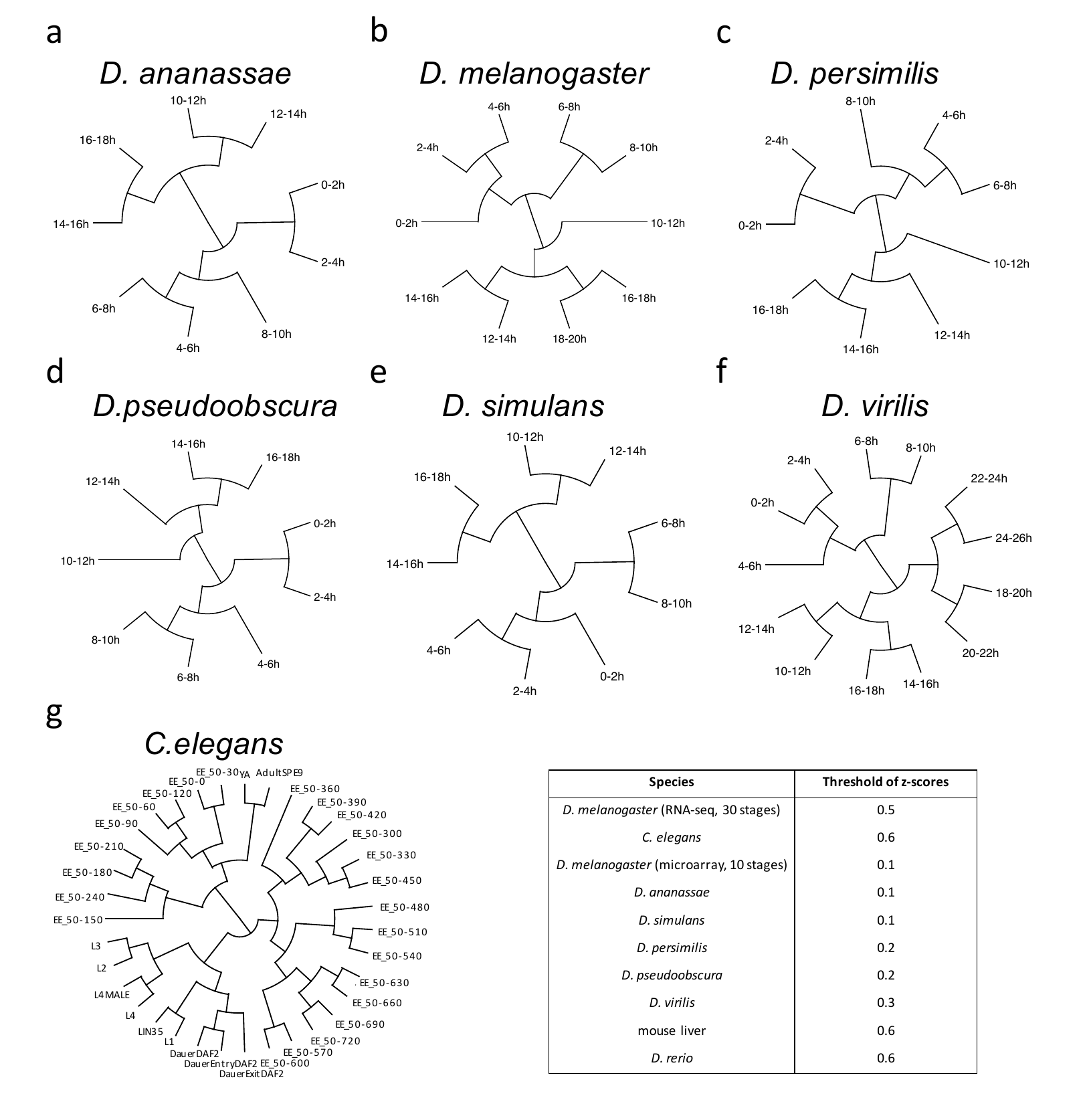}}
\caption{Developmental tree constructed using stage-associated genes (identified with the $z$-score thresholds in the table). \textbf{a}-\textbf{f} are for \emph{Drosophila} species and \textbf{g} is for \emph{C.elegans}.}
\label{tree_supp}
\end{figure}

\begin{figure}[!pb]
\centerline{\includegraphics[width=\linewidth]{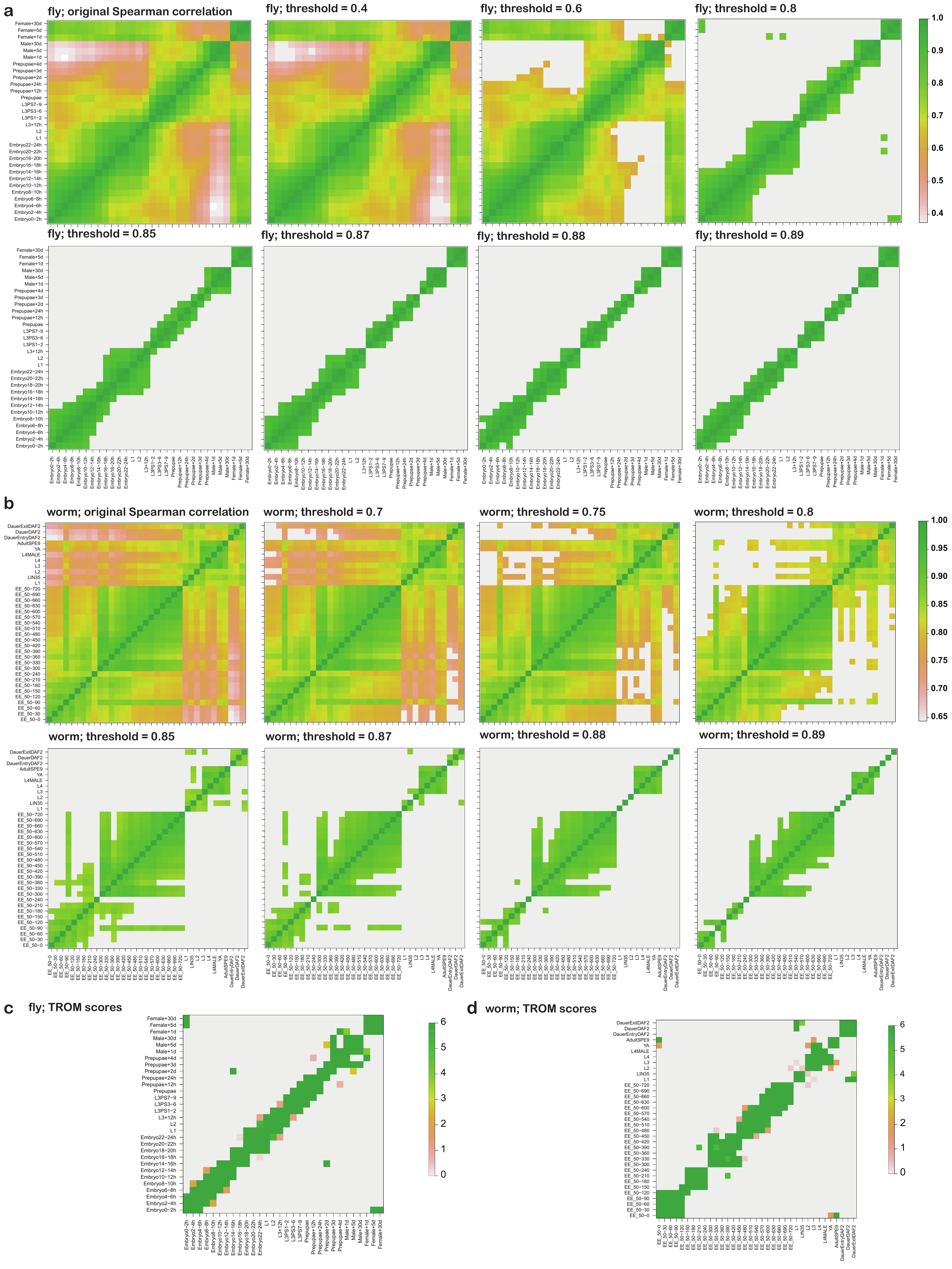}}
\cprotect\caption{Spearman correlation of the developmental stages of \emph{D. melanogaster} (fly) and \emph{C. elegans} (worm). \textbf{a}: The first panel shows Spearman correlation of fly's stages while the rest show Spearman correlation of fly's stages under different thresholds.  \textbf{b}: The first panel shows Spearman correlation of worm's stages while the rest show Spearman correlation of worm's stages under different thresholds. \textbf{c}: TROM scores of fly. \textbf{d}: TROM scores of worm. All the values under the selected threshold are set to the minimum value of each correlation matrix.}
\label{sp_thre}
\end{figure}

\end{document}